\documentclass[acmsmall]{acmart}

\usepackage{prelude}
\usepackage{tkz-graph}

\usetikzlibrary{shapes.geometric, arrows}

\newtoggle{widecol}
\toggletrue{widecol}

\newcommand{\captiondetail}{%
  \justifying\small\mbox{}\\[-\baselineskip]\noindent}
\newcommand{\nocaptiondetail}{\vspace{-0.75\baselineskip}}
\newcommand{\squish}{\iftoggle{widecol}{}{\vspace{-0.75\baselineskip}}}

\newcommand{\proofin}[2][Proof]{%
  \begin{proof}[#1] See \fref{#2}. \noqed \end{proof}}
\newcommand{\proofof}[1]{Proof of \fref{#1}}

\newcommand{\integral}[2]{\mkern-3.5mu\int\displaylimits_{#1}^{#2}\mkern-3.5mu}
\newcommand{\nudgel}{\mkern-1.75mu\cdot\mkern-3.25mu}
\newcommand{\nudgec}{\mkern-2.5mu\cdot\mkern-2.5mu}

\newcommand{\tail}[2][]{\appopt{\subopt{\ol{F}}{#1}}{#2}}
\newcommand{\density}[2][]{\appopt{\subopt{f}{#1}}{#2}}

\newcommand{\stages}[1][]{\subopt{\mc{I}}{#1}}
\newcommand{\size}[1]{\subopt{X}{#1}}
\newcommand{\totalsize}[1]{\subopt{X}{#1}}
\newcommand{\trans}[3][]{\subopt{\supopt{p}{#1}}{#2#3}}
\newcommand{\init}[1][]{\subopt{a}{#1}}
\newcommand{\fin}[1][]{\subopt{z}{#1}}
\newcommand{\st}[3][]{#1\angle{#2, #3}}
\newcommand{\sts}{\vec{\imath x}}
\newcommand{\hazard}[2]{\appopt{h_{#1}}{#2}}

\newcommand{\fair}[1]{\subopt{M}{#1}}
\newcommand{\gittins}[1]{\subopt{G}{#1}}
\newcommand{\prevailing}[2]{R_{\st{#1}{#2}}}
\newcommand{\freezing}[2]{S_{\st{#1}{#2}}}
\newcommand{\profit}[2]{\appopt{\subopt{V}{#1}}{#2}}
\newcommand{\profitinv}[2]{\appopt{\subopt{V^{-1}}{#1}}{#2}}

\newcommand{\jt}[1]{\texttt{\scalebox{1.04}{#1}}}
\newcommand{\J}{\jt{J}}
\newcommand{\K}{\jt{K}}
\renewcommand{\L}{\jt{L}}
\newcommand{\single}[1]{[#1]}
\newcommand{\seq}{\rhd}

\newcommand{\Tqueue}{T_{\text{Q}}}

\newcommand{\load}[1]{\sqappopt{\rho}{#1}}
\newcommand{\frozen}[1]{s[#1]}
\newcommand{\cost}[2][]{\appopt{\subopt{\phi}{#1}}{#2}}
\newcommand{\costA}[2]{\appopt{\sqappopt{\Phi_\text{A}}{#1}}{#2}}
\newcommand{\costB}[2]{\appopt{\sqappopt{\Phi_\text{B}}{#1}}{#2}}
\newcommand{\costC}[2]{\appopt{\sqappopt{\Phi_\text{C}}{#1}}{#2}}
\newcommand{\costP}[1]{\appopt{\Phi_\text{P}}{#1}}
\newcommand{\costQ}[1]{\appopt{\Phi_\text{Q}}{#1}}

\newcommand{\bypass}[2]{\sqappopt{\subopt{\Gamma}{#2}}{#1}}
\newcommand{\serve}[1]{\Delta_{#1}}
\newcommand{\arrive}{\Lambda}

\acmclean

\begin{document}

\title{%
  Optimal Scheduling and Exact Response Time Analysis for Multistage Jobs}

\author{Ziv Scully}
\affiliation{%
  \institution{Carnegie Mellon University}
  \department{Computer Science Department}
  \streetaddress{5000 Forbes Ave}
  \city{Pittsburgh}
  \state{PA}
  \postcode{15213}
  \country{USA}}
\email{zscully@cs.cmu.edu}

\author{Mor Harchol-Balter}
\affiliation{%
  \institution{Carnegie Mellon University}
  \department{Computer Science Department}
  \streetaddress{5000 Forbes Ave}
  \city{Pittsburgh}
  \state{PA}
  \postcode{15213}
  \country{USA}}
\email{harchol@cs.cmu.edu}

\author{Alan Scheller-Wolf}
\affiliation{%
  \institution{Carnegie Mellon University}
  \department{Tepper School of Business}
  \streetaddress{5000 Forbes Ave}
  \city{Pittsburgh}
  \state{PA}
  \postcode{15213}
  \country{USA}}
\email{awolf@andrew.cmu.edu}

\begin{abstract}
  Scheduling to minimize mean response time
  in an M/G/1 queue is a classic problem.
  The problem is usually addressed in one of two scenarios.
  In the perfect-information scenario,
  the scheduler knows each job's exact size, or service requirement.
  In the zero-information scenario,
  the scheduler knows only each job's size distribution.
  The well-known shortest remaining processing time (SRPT) policy is optimal
  in the perfect-information scenario,
  and the more complex Gittins policy is optimal
  in the zero-information scenario.

  In real systems the scheduler often has \emph{partial} but incomplete
  information about each job's size.
  We introduce a new job model, that of \emph{multistage jobs},
  to capture this partial-information scenario.
  A multistage job consists of a sequence of \emph{stages},
  where both the sequence of stages and stage sizes are unknown,
  but the scheduler always knows \emph{which stage} of a job is in progress.
  We give an optimal algorithm for scheduling multistage jobs in an M/G/1 queue
  and an exact response time analysis of our algorithm.
\end{abstract}

\maketitle

\section{Introduction}

Scheduling jobs to minimize mean response time
in preemptive single-server queueing systems,
such as the M/G/1 queue,
has been the subject of countless papers over the past several decades.
When the scheduler knows each job's exact \emph{size},
or service requirement,
the \emph{shortest remaining processing time} (SRPT) algorithm is optimal.
When the scheduler does not know jobs' sizes,
the optimal policy is the \emph{Gittins policy}
\citep{book_gittins, m/g/1_gittins_aalto},
which uses the job size distribution to
give each job a priority based on its \emph{age},
namely the amount of service the job has received so far.
Although both SRPT and the Gittins policy
have long been known to be optimal for their respective settings,
only SRPT has been analyzed in the past \citep{srpt_analysis_schrage}.
There is no known closed form for mean response time
under the Gittins policy.

In some sense, SRPT and the Gittins policy
represent opposite extremes of optimal scheduling algorithms.
SRPT treats the complete information case:
we know each individual job's exact size.
The Gittins policy treats the zero information case:
we have no information about individual jobs,
leaving us only the overall size distribution
and each job's age on which to base a policy.
This prompts a natural question:
what policy is optimal in the \emph{partial information} case,
where we have nonzero but incomplete information
about each job's remaining size?

To model jobs with partial information about remaining size,
we propose a new job model: \emph{multistage jobs}.
During service,
a multistage job progresses through a sequence of \emph{stages},
each of which has its own individual \emph{size}, or service requirement.
A job completes when its last stage completes.
The sequence of stages a job goes through may be stochastic,
and the server cannot influence the sequence of stages.
Whereas a standard M/G/1 with unknown job sizes
uses only the age of each job when scheduling,
multistage jobs provide the scheduler both
\begin{itemize}
\item
  the \emph{current stage} of the job and
\item
  the \emph{age of the current stage}.
\end{itemize}
Although stages may have general size distributions,
stage sizes and transitions are Markovian in the sense that
a job's history before its current stage
is independent of the sizes of and transitions between
its current and future stages.

Jobs with multiple stages are so common in applications
that it is surprising they have not received more attention.
For instance, the following scenarios can be modeled with multistage jobs:
\begin{itemize}
\item
  At every moment in time,
  a Google ad server preemptively chooses which ad placement job to serve.
  Ad placement has three stages:
  preprocessing, which has uniform size distribution;
  targeting, which has bounded Pareto size distribution with a low bound;
  and selection,
  which has bounded Pareto size distribution with a high bound.\footnote{%
    Personal communication.}
  Job~$\jt{G}$ in \fref{fig:multistage_examples} models an ad placement request
  as a multistage job.
\item
  In a small general repair shop,
  a single repairer preemptively chooses which of multiple items to fix.
  Each item goes through two stages of service: diagnosis and repair.
  The amount of time repair takes depends on the problem with the item,
  which the repairer knows only after diagnosis.
  Job~$\jt{R}$ in \fref{fig:multistage_examples} models this repair process
  as a multistage job.
\item
  Drugs are developed through multiple stages of research and trials.
  A single team preemptively chooses
  which of multiple drug projects to work on
  based on the state of each project.
\end{itemize}

\begin{figure}
  \centering
  \begin{tikzpicture}[scale=0.8, thick]
    \newcommand{\leftlbl}{left, fill=none, shift={(0,0.1)}}
    \newcommand{\rightlbl}{right, fill=none, shift={(0,0.1)}}
    \SetUpVertex[Math,LabelOut]
    \tikzstyle{EdgeStyle}=[->, >=stealth']
    \tikzstyle{VertexStyle}=[shape=circle, minimum width=1em, draw]
    \begin{scope}[shift={(-1.5,0)}]
      \node at (0, 0.7) {$\jt{G}$};
      \Vertex{P}
      \SO(P){T}
      \SO(T){S}
      \Edges(P,T,S)
    \end{scope}
    \begin{scope}[shift={(3,-0.5)}]
      \node at (0, 1.2) {$\jt{R}$};
      \Vertex[Lpos=90]{D}
      \SOWE[L={R_{\text{easy}}}, Lpos=-90](D){R_e}
      \SOEA[L={R_{\text{hard}}}, Lpos=-90](D){R_h}
      \Edge[label=$p$, labelstyle={\leftlbl}, style={pos=0.4}](D)(R_e)
      \Edge[label=$1 - p$, labelstyle={\rightlbl}, style={pos=0.4}](D)(R_h)
    \end{scope}
  \end{tikzpicture}

  \captiondetail
  Two examples of multistage jobs.
  Job~$\jt{G}$ is a Google ad placement request,
  which always goes through the same three stages before completion.
  The first stage, preprocessing, has uniform size distribution~$P$.
  The other stages, targeting and selection,
  have bounded Pareto size distributions $T$ and $S$
  with low and high upper bound, respectively.
  Job~$\jt{R}$ is a job in a general repair shop.
  The first stage is diagnosis, which takes time~$D$.
  The other stages are repair work.
  With probability~$p$, repair work is easy and takes time~$R_{\text{easy}}$,
  and with probability $1 - p$,
  repair work is hard and takes time~$R_{\text{hard}}$.
  \squish

  \caption{Multistage Job Examples}
  \label{fig:multistage_examples}
\end{figure}

To highlight the difference between multistage jobs and previous job models,
consider multistage jobs $\jt{A}$, $\jt{B}$, and $\jt{C}$
in \fref{fig:multistage_difference}.
While all three jobs have total size distribution $S + 2d$,
where $S$ is stochastic and $d$ is deterministic,
their stage structures make them very different.
\begin{itemize}
\item
  In job~$\jt{A}$,
  the stochastic stage is \emph{after} the deterministic stage:
  after serving $\jt{A}$ for time~$2d$,
  we know the remaining size has distribution~$S$.
\item
  In job~$\jt{B}$,
  the stochastic stage is \emph{before} the deterministic stage:
  after serving $\jt{B}$ for a random amount of time~$S$,
  we learn that the first stage has completed,
  at which point we know the remaining size is exactly~$2d$.
\item
  Job~$\jt{C}$ combines both of these effects:
  after serving $\jt{C}$ for time~$d$,
  we then serve it for a random amount of time~$S$
  until we learn that the second stage has completed,
  at which point the remaining size is~$d$.
\end{itemize}

Suppose we have a preemptive single-server system
containing jobs $\jt{A}$, $\jt{B}$, and~$\jt{C}$.
Which job should we serve first to minimize mean response time?
Is it better to serve $\jt{A}$ first
to get the deterministic section out of the way?
Is it better to serve $\jt{B}$ first,
just in case we get lucky and its first stage finishes quickly?
Is the compromise $\jt{C}$ a better or worse choice than each of the others?
\begin{quote}
  How do we best \emph{exploit multistage structure} when scheduling
  to minimize mean response time?
\end{quote}
This is a difficult question to answer because it requires comparing
not just jobs' current size distributions
but also their \emph{potential future size distributions}.
The differences between $\jt{A}$, $\jt{B}$, and $\jt{C}$
lie not in their initial size distributions,
all $S + d$,
but rather in how their size distributions might evolve with service.
Adding to the difficulty is the fact that
a multistage job's remaining size distribution evolves \emph{stochastically},
jumping every time a stage completes.
This is more complicated than a single-stage job in the standard M/G/1,
whose remaining size distribution depends only on its age.
Jobs with stochastic stage sequences,
such as job~$\jt{R}$ in \fref{fig:multistage_examples},
exacerbate the difficulty.

\begin{figure}
  \centering
  \begin{tikzpicture}[scale=0.8, thick]
    \SetUpVertex[Math,LabelOut]
    \tikzstyle{EdgeStyle}=[->, >=stealth']
    \tikzstyle{VertexStyle}=[shape=circle, minimum width=1em, draw]
    \begin{scope}[shift={(-3,0)}]
      \node at (0, 0.7) {$\jt{A}$};
      \Vertex[L={d}]{d1}
      \SO[L={d}](d1){d2}
      \SO(d2){S}
      \Edges(d1,d2,S)
    \end{scope}
    \begin{scope}[shift={(0,0)}]
      \node at (0, 0.7) {$\jt{B}$};
      \Vertex{S}
      \SO[L={d}](S){d1}
      \SO[L={d}](d1){d2}
      \Edges(S,d1,d2)
    \end{scope}
    \begin{scope}[shift={(3,0)}]
      \node at (0, 0.7) {$\jt{C}$};
      \Vertex[L={d}]{d1}
      \SO(d1){S}
      \SO[L={d}](S){d2}
      \Edges(d1,S,d2)
    \end{scope}
  \end{tikzpicture}

  \captiondetail
  Let $S$ be a stochastic size distribution
  and $d$ be a deterministic size.
  The three jobs $\jt{A}$, $\jt{B}$, and $\jt{C}$
  all have the same total size distribution $S + 2d$,
  but they have different stage structures.
  \squish

  \caption{Different Stage Structures for Same Total Size}
  \label{fig:multistage_difference}
\end{figure}

To our knowledge, there is \emph{no prior theoretical work}
specifically pertaining to multistage jobs in a preemptive setting.
Analyzing the mean response time of
virtually any scheduling policy for multistage jobs\footnote{%
  Here we mean a scheduling policy that makes some use of multistage structure.
  One can, of course, analyze classic scheduling policies
  by simply ignoring stage information.}
is an open problem,
let alone finding and analyzing the optimal policy.
The closest prior work is
\emph{Klimov's model} and variations
\citep{klimov's_model_klimov, tax_undiscounted_whittle},
which can model jobs with nonpreemptible stages,
whereas our stages are preemptible.
Nonpreemptible stages require far fewer scheduling decisions,
which makes Klimov's model simple to optimally schedule
using a variant of the Gittins policy
\citep{branching_bandits_bertsimas, multiple_processors_weiss, book_gittins}.

The main contribution of this paper is
an \emph{optimal scheduling algorithm} for multistage jobs.
Thanks to a \emph{novel reduction technique},
our algorithm is no more complex than
optimally scheduling single-stage jobs.
We \emph{exactly analyze the performance} of our algorithm,
obtaining a closed-form expression mean response time.

\subsection{Background and Challenges}
\label{sub:challenges}

As mentioned above,
we will reduce the problem of scheduling multistage jobs
to the problem of scheduling single-stage jobs with unknown size,
as in the standard M/G/1.
It thus behooves us to review the single-stage setting,
in which the \emph{Gittins policy} minimizes mean response time
\citep{book_gittins, m/g/1_gittins_aalto}.
Remarkably, the Gittins policy has a very simple form:
for each job,
determine a quantity called its \emph{Gittins index},
then serve the job of maximal Gittins index.
While the Gittins policy has long been known to be optimal,
its response time resisted analysis until recently,
when \citet{soap_scully} analyzed the class of \emph{SOAP policies},
which includes the Gittins policy.

One might wonder whether we can apply the Gittins index to multistage jobs.
The answer is yes, but with a caveat:
it is possible to define the Gittins index of a multistage job,
and the resulting Gittins policy minimizes mean response time,
but \emph{computing} a multistage job's Gittins index requires solving
a multidimensional optimization problem with one dimension per stage.
This makes computing Gittins indices intractable.
We might hope to tame this complexity by, for instance,
computing the Gittins index of a multistage job
by somehow combining the Gittins indices of its individual stages,
which would allow us to leverage prior work on single-stage jobs,
but there is no known way to do this.

Even if we could compute the Gittins policy for multistage jobs,
analyzing its mean response time would still be challenging.
The analysis of SOAP policies given by \citet{soap_scully}
is for single-stage jobs.
While a similar method could in principle be used for multistage jobs,
it would require considering all the possible sequences of stages,
of which there can be exponentially many.
No known analysis method scales gracefully to multistage jobs.

\subsection{Contributions}

In this paper,
we solve the two main challenges outlined in \fref{sub:challenges}:
we give a simple method to compute the Gittins index of a multistage job,
and we give a simple closed form for
the mean response time of the Gittins policy for multistage jobs.
Specifically, our contributions are as follows:
\begin{itemize}
\item
  We introduce \emph{multistage jobs},
  a new job model in which the scheduler has partial information
  about a job's progress (\fref{sec:model}).
\item
  We introduce \emph{single-job profit} (SJP),
  a new framework for the Gittins index
  that is especially well suited for multistage jobs (\fref{sec:sjp}).
\item
  Using the SJP framework, we state and prove a new \emph{composition law},
  which reduces a multistage job's Gittins index computation
  into two smaller SJP computations which,
  unlike Gittins indices, naturally compose (\fref{sec:composition}).
  We give two applications of the composition law:
  \begin{itemize}
  \item
    We show how repeated application of the composition law
    reduces any multistage job's Gittins index computation
    to SJP computations for \emph{single-stage jobs} (\fref{sub:single-stage}).
  \item
    In special case of a fixed stage sequence with discrete size distributions,
    the problem simplifies considerably,
    and our composition law provides a divide-and-conquer algorithm
    for computing the Gittins index in $O(n \log n)$ time
    where $n$ is the total number of support points (\fref{sub:discrete}).
  \end{itemize}
\item
  Using the SJP framework,
  we \emph{exactly analyze mean response time} under our algorithm,
  giving a closed form in terms of SJP (\fref{sec:analysis}).
\item
  We give several \emph{practical takeaways}
  for systems with multistage jobs (\fref{sec:takeaways}):
  \begin{itemize}
  \item
    We demonstrate that, as a rough guideline,
    it is best to \emph{front-load stages with more variable length}
    (\fref{sub:prioritize_variable}).
    This answers the question posed by \fref{fig:multistage_difference}.
  \item
    We show that our algorithm \emph{significantly reduces response time}
    compared to all algorithms that do not exploit multistage structure
    (\fref{sub:impact}).
  \end{itemize}
\end{itemize}
Together, our contributions constitute
the \emph{first theoretical analysis} of
multistage jobs with \emph{preemptible} stages
of \emph{general size distribution}.

\subsection{Related Work}

The Gittins index has been extensively studied in the standard M/G/1 setting
\citep{m/g/1_gittins_aalto, mlps_gittins_aalto, multiclass_ayesta,
  rs_slowdown_hyytia},
but there are few concrete results for more complex job models.
A notable exception is Klimov's model \citep{klimov's_model_klimov},
which is a nonpreemptive multiclass M/G/1 queue
with Bernoulli feedback:
upon completion, jobs of class~$i$ have some probability $p_{ij}$
of immediately returning as a class~$j$ job.
By thinking of classes as stages,
we can think of a job in Klimov's model as a job with multiple stages.
One can optimally schedule Klimov's model using the Gittins index
\citep[Section~4.9.2]{book_gittins}.
The fundamental difference between Klimov's model
and our model is that we allow \emph{preemption},
which makes the space of possible policies much larger
and makes Gittins index computation far more difficult.
A variation of Klimov's model studied by \citet{tax_undiscounted_whittle}
allows preemption
but requires each class's size distribution to be exponential,
whereas we study \emph{general size distributions}.

The Gittins index has been applied to many problems
in scheduling and beyond
\citep{book_gittins, four_proofs_weiss, original_gittins, review_gittins,
  gittins_index_weber, gittins_index_whittle, golf_dumitriu,
  achievable_region_bertsimas, tax_undiscounted_whittle}.
As such, there are many definitions of the Gittins index,
all essentially equivalent but each with different strengths,
such as providing intuition or enabling certain elegant proofs.
SJP provides another definition of the Gittins index
which is particularly well suited to working with multistage jobs.
See \fref{sub:prior_gittins} for a more detailed comparison
between SJP and prior Gittins index formulations.

\section{System Model}
\label{sec:model}

We consider the \emph{multistage M/G/1 queue},
a preemptive single-server queueing system in which
the jobs are \emph{multistage jobs}.
A multistage job goes through
a stochastic sequence of \emph{stages} during service.
Each stage has a \emph{size}, or service requirement.

In the multistage M/G/1,
jobs arrive as a Poisson process of rate~$\lambda$,
and each job's stage sizes and transitions
are independent of those of other jobs and the arrival process.
At every moment in time, the scheduler must decide which job to serve.
The scheduler does not know jobs' exact stage sequences
or stages' exact sizes ahead of time,
but it does know the distributions of possible sequences and sizes
(\fref{sub:anatomy}).
Our goal is to minimize \emph{mean response time},
the average time between a job's arrival and completion.
We assume a stable system,
meaning the expected total size of each job is less than~$1/\lambda$,
and a preempt-resume model,
meaning preemption and processor sharing are permitted
with no penalty or loss of work.

\subsection{Anatomy of a Multistage Job}
\label{sub:anatomy}

In this section we zoom in on a single job
and define its dynamics during service,
summarizing the resulting notation in \fref{tab:job_notation}.

Every multistage job is a pair of two pieces of data:
a \emph{type} and a \emph{state}.
We can think of a job's type as a ``map''
and a job's state as a ``location'' on the map.
As a job is served,
it's state evolves according to transition rules based on the job's type.
The scheduler knows the type and state of every job currently in the system.

A job \emph{type}~$\J$ specifies the following information,
which is all known to the scheduler.
\begin{itemize}
\item
  Type~$\J$ jobs have a finite set of possible \emph{stages},
  denoted~$\stages[\J]$.
  A stage is simply a label or name.
  Below we associate transition probabilities and a size distribution
  with each label.
\item
  Each type~$\J$ job starts in the \emph{initial stage}
  $\init[\J] \in \stages[\J]$.
\item
  There is a special \emph{final stage} $\fin[\J] \in \stages[\J]$
  that each type~$\J$ job enters upon completion.
\item
  Upon completing stage~$i$,
  a type~$\J$ job transitions to stage~$j$
  with \emph{transition probability}~$\trans[\J]{i}{j}$.
  The stages form an acyclic Markov chain
  absorbing at~$\fin[\J]$,
  meaning~$\fin[\J]$ is accessible from all other stages
  and no stage is accessible from itself.\footnote{%
    One could certainly model jobs with cyclic stage transitions,
    and many of our results still hold in this case,
    but we assume acyclic transitions for simplicity.}
\item
  Each stage $i \in \stages[\J]$ has a \emph{size distribution}~$\size{i}$
  representing the amount of service stage~$i$ requires.
  We write $\density[i]{}$ and $\tail[i]{}$
  for the density\footnote{%
    For simplicity of notation, we assume the density always exists,
    but our results generalize to size distributions with discrete components.}
  and tail functions, respectively, of~$\size{i}$.
  We always have $\size{\fin[\J]} = 0$.
\end{itemize}

A job's \emph{state}, denoted $\st{i}{x}$,
consists of the stage~$i$ in progress
and the \emph{age} $x \geq 0$ of the stage,
namely the amount of time the job has been served while in stage~$i$.
Stage transitions and stage sizes are independent,
so the evolution of a job's state $\st{i}{x}$ during service is Markovian,
even for generally distributed stage sizes~$\size{i}$.
When a type~$\J$ job is served in state $\st{i}{x}$,
the state experiences two kinds of transitions:
\begin{itemize}
\item
  The stage's age~$x$ increases continuously at rate~$1$.
\item
  The job's state stochastically jumps from $\st{i}{x}$ to $\st{j}{0}$
  at instantaneous rate $\trans[\J]{i}{j}\hazard{i}{x}$, where
  \begin{equation*}
    \hazard{i}{x} = \frac{\density[i]{x}}{\tail[i]{x}}
  \end{equation*}
  is the \emph{hazard rate} of stage size~$\size{i}$ at age~$x$.
\end{itemize}
A type~$\J$ job starts in state $\st{\init[\J]}{0}$ and
\emph{completes}, exiting the system,
upon jumping to state $\st{\fin[\J]}{0}$.
We write $\E{\totalsize{\J}}$ for the expected \emph{total size}
of a type~$\J$ job,
which is the expected amount of service time required for a type~$\J$ job
to go from $\st{\init[\J]}{0}$ to $\st{\fin[\J]}{0}$.

\begin{table}
  \caption{Summary of Multistage Job Notation}
  \label{tab:job_notation}

  \squish\begin{tabular}{@{}ll@{}}
    \toprule
    \emph{Notation} & \emph{Description} \\
    \midrule
    $\stages[\J]$ & set of possible stages of job type~$\J$ \\
    $\init[\J], \fin[\J]$ & initial and final stages of job type~$\J$ \\
    $\trans[\J]{i}{j}$
    & transition probability from $i$ to~$j$ in job type~$\J$ \\
    $\size{i}$ & size distribution of stage~$i$ \\
    $\st{i}{x}$ & state of a job at age~$x$ of stage~$i$ \\
    $\density[i]{}, \tail[i]{}, \hazard{i}{}$
    & density, tail, and hazard rate functions of $\size{i}$ \\
    $\E{\totalsize{\J}}$ & expected total size of~$\J$ \\
    $\st[\J]{i}{x}$
    & state-conditioned job type \\
    \bottomrule
  \end{tabular}\squish
\end{table}

\subsection{Terminology Conventions}
\label{sub:conventions}

In Sections~\ref{sec:sjp} and~\ref{sec:composition}
we study single-job profit,
a framework for computing the Gittins index of a multistage job.
To simplify exposition,
we assume that all jobs are in their initial state,
and when we say ``job~$\J$'' we mean a type~$\J$ job
in its initial state $\st{\init[\J]}{0}$.
To apply our results to a job in any state,
we can simply define a new job type.
Specifically, given job type~$\J$ and state $\st{i}{x}$,
the \emph{state-conditioned job type} $\st[\J]{i}{x}$
is the job type obtained by starting a type~$\J$ job in state $\st{i}{x}$.
That is, a type $\st[\J]{i}{x}$ job in state $\st{\init[{\st[\J]{i}{x}}]}{0}$
has future evolution stochastically equivalent to that of
a type~$\J$ job in state $\st{i}{x}$.
See \fref{app:definitions} for a formal definition.

In \fref{sec:analysis} we study the M/G/1 with multistage jobs.
To simplify notation,
we assume that all jobs have the same job type~$\J$
and omit the $\J$ from the notation in \fref{tab:job_notation}.
Despite using a single job type,
we can model a system with multiple job classes
with what we call a \emph{mixture composition}
(\fref{def:mixture}):
we give the initial stage~$a$ size $\size{a} = 0$,
causing each new job to immediately transition to what we think of as
the initial stage for its class.

\section{Single-Job Profit: A New Framework for the Gittins Index}
\label{sec:sjp}

In this section we describe \emph{single-job profit} (SJP),
a new framework for defining and computing the Gittins index of a job.
As we will see, SJP is particularly well suited to multistage jobs.
We first define SJP (\fref{sub:sjp_definition})
then relate it to the Gittins index (\fref{sub:sjp_gittins}).
Throughout this section we define several terms and notations related to SJP,
which we summarize in \fref{tab:sjp_notation}.

\subsection{Definition of Single-Job Profit}
\label{sub:sjp_definition}

We define \emph{single-job profit} (SJP) to be
an optimization problem concerning
a single multistage job~$\J$ and a potential \emph{reward}~$r \geq 0$.
At every moment in time,
we must choose between one of two actions:
\emph{serving} and \emph{giving up}.
While serving,
we incur cost at continuous rate~$1$
as $\J$'s state evolves as described in \fref{sub:anatomy}.
If $\J$ completes, we receive reward~$r$ and the process ends.
The other action, giving up,
immediately ends the process with no additional cost or reward.
Our goal is to maximize \emph{expected profit},
meaning the expected reward received minus expected cost incurred.

Because $\J$'s state evolution is Markovian,
to play optimally in SJP,
it suffices to consider policies that decide to serve or give up
based only on the current state.
Thus, a policy for SJP is a \emph{stopping policy}~$\pi$
consisting of a \emph{stopping age} $\pi(i) \geq 0$
for each stage~$i$.
A stopping policy~$\pi$ gives up if $\J$ enters state $\st{i}{\pi(i)}$
for some stage~$i$.
For example, the policy of never giving up
sets $\pi(i) = \infty$ for all stages~$i$,
and the policy of giving up immediately sets $\pi(\init[\J]) = 0$.

\begin{definition}
  \label{def:sjp}
  The \emph{optimal profit} of job~$\J$ and reward~$r$ in SJP is
  \begin{equation*}
    \profit{\J}{r}
    = \sup_\pi
      (r\P{\text{$\pi$ completes~$\J$}}
        - \E{\text{time $\pi$ serves~$\J$}}),
  \end{equation*}
  where the supremum is taken over stopping policies~$\pi$.
  The optimal profit as a function of reward~$r$
  is called $\J$'s \emph{SJP function}.
\end{definition}

It is intuitive that a state's SJP function is nondecreasing,
because a larger potential reward should not hurt the optimal profit.
In fact, we can make a much stronger statement.

\begin{lemma}
  \label{lem:sjp_properties}
  Every job's SJP function $\profit{\J}{r}$ is
  \begin{enumerate}[(i)]
  \item
    \label{item:convex_nondecreasing}
    convex and nondecreasing,
  \item
    \label{item:bounded_above}
    bounded above by~$r$,
  \item
    \label{item:bounded_below}
    bounded below by $\max\{0, r - \E{\totalsize{\J}}\}$,
    and
  \item
    \label{item:bounded_derivative}
    differentiable almost everywhere with derivative at most~$1$.
  \end{enumerate}
\end{lemma}

\begin{proof}
  From \fref{def:sjp},
  we see $\profit{\J}{r}$ is a supremum of
  convex nondecreasing functions of~$r$,
  implying claim~\fref{item:convex_nondecreasing}.
  No policy can receive reward greater than~$r$ or incur cost less than~$0$,
  implying claim~\fref{item:bounded_above}.
  Possible policies for SJP include immediately giving up,
  which yields expected profit~$0$,
  and never giving up,
  which yields expected profit $r - \E{\totalsize{\J}}$,
  implying claim~\fref{item:bounded_below}.
  The conjunction
  of claims \fref{item:convex_nondecreasing} and~\fref{item:bounded_above}
  implies~\fref{item:bounded_derivative}.
\end{proof}

By \fref{lem:sjp_properties},
$\J$'s SJP function has an inverse $\profitinv{\J}{u}$ defined for all $u > 0$.
There are many possibilities for $\profitinv{\J}{0}$.
We choose the value that makes $\profitinv{\J}{u}$ continuous at $u = 0$,
which is the largest possible value.

\begin{definition}
  \label{def:fair}
  The \emph{SJP index} of job~$\J$ is
  \begin{align*}
    \fair{\J}
    &= \profitinv{\J}{0} \\
    &= \sup\{r \geq 0 \mid \profit{\J}{r} = 0\} \\
    &= \inf\{r \geq 0 \mid \profit{\J}{r} > 0\}.
  \end{align*}
\end{definition}

Put another way, the SJP index of job~$\J$
is the smallest reward~$r$ such that
it is optimal to serve $\J$ for at least an instant in SJP.
A state's SJP index is a measure of
how reluctant we are to serve a job in that state:
the more reluctant we are to serve the job,
the larger a reward we need to be offered to serve it.

\begin{table}
  \caption{Summary of Single-Job Profit Notation}
  \label{tab:sjp_notation}

  \centering
  \squish\begin{tabular}{@{}lll@{}}
    \toprule
    \emph{Notation} & \emph{Description} & \emph{Defined in} \\
    \midrule
    $\profit{\J}{r}$
      & SJP function of job~$\J$
      & \fref{def:sjp} \\
    $\fair{\J} = \profitinv{\J}{0}$
      & SJP index of job~$\J$
      & \fref{def:fair} \\
    $\gittins{\J} = 1/\fair{\J}$
      & Gittins index of job~$\J$
      & \fref{def:gittins} \\
    \bottomrule
  \end{tabular}\squish
\end{table}

\subsection{Single-Job Profit and the Gittins Index}
\label{sub:sjp_gittins}

We now relate SJP to the traditional definition of the Gittins index,
beginning with a review of the traditional definition.

The \emph{Gittins policy} is a policy for minimizing mean response time
in a variety of queueing systems,
including the standard M/G/1 \citep{m/g/1_gittins_aalto, mlps_gittins_aalto}
and the multistage M/G/1 considered herein.
The policy has a simple form:
always serve the job of maximal \emph{Gittins index}.
A job's Gittins index, defined below,
is unaffected by other jobs in the system.

\begin{definition}[{\citep[Section~3.2]{book_gittins}}]
  \label{def:gittins}
  The \emph{Gittins index} of job~$\J$ is
  \begin{equation*}
    \gittins{\J}
    = \sup_\pi \frac{%
        \P{\text{$\pi$ completes~$\J$}}}{%
        \E{\text{time $\pi$ serves~$\J$}}},
  \end{equation*}
  where the supremum is taken over stopping policies~$\pi$
  that do not immediately give up.
\end{definition}

One interpretation of $\gittins{\J}$
is as the maximum possible ``completion rate'' of~$\J$:
the numerator in \fref{def:gittins} is an expected number of job completions,
and the denominator is the expected time for those completions to occur.
The significance of Gittins indices lies in the following theorem.

\begin{theorem}[{\citep[Corollary~3.5]{book_gittins}}]
  \label{thm:gittins}
  The Gittins policy,
  which always serves the job of greatest Gittins index,
  minimizes mean response time in any single-server system with
  Poisson job arrivals and Markovian job dynamics.
\end{theorem}

The Gittins index theorem applies to our multistage M/G/1,
but directly computing the Gittins indices of multistage job
using \fref{def:gittins}
requires optimizing the multidimensional parameter~$\pi$.
Fortunately,
we can obtain a multistage job's Gittins index from SJP
as $\gittins{\J} = 1/\fair{\J}$.
This helps because, as we show in \fref{sub:single-stage},
SJP with a multistage job reduces to
several SJP instances with single-stage jobs.

Where does $\gittins{\J} = 1/\fair{\J}$ come from?
Consider SJP with job~$\J$ and reward~$r$ and ask:
is it optimal to give up immediately?
We can answer in two ways:
\begin{itemize}
\item
  \emph{Using the SJP index:}
  By \fref{def:fair},
  it is optimal to give up immediately if $r \leq \fair{\J}$.
\item
  \emph{Using the Gittins index:}
  Recall that we can think of $\gittins{\J}$ as
  the maximum possible ``completion rate'' of~$\J$.
  This means that with the right stopping policy,
  we can receive reward at average rate $r\gittins{\J}$
  while incurring cost at rate~$1$,
  so it is optimal to give up immediately if $r\gittins{\J} \leq 1$.
\end{itemize}
The following proposition formalizes this argument.

\begin{proposition}
  \label{prop:gittins_fair}
  For any job~$\J$,
  \begin{equation*}
    \fair{\J} = \frac{1}{\gittins{\J}}.
  \end{equation*}
\end{proposition}

\begin{proof}
  Throughout, $\pi$ ranges over
  stopping policies that do not give up immediately.
  Continuity ensures the supremum in \fref{def:sjp}
  is unchanged by omitting the policy of giving up immediately.
  Let
  \begin{align*}
    P_\pi &= \P{\text{$\pi$ completes~$\J$}} \\
    E_\pi &= \E{\text{time $\pi$ serves~$\J$}}.
  \end{align*}
  We compute
  \iftoggle{widecol}{\begin{equation*}}{\begin{align*}}
    \fair{\J}
    \iftoggle{widecol}{}{&}
    = \inf\{r \geq 0 \mid \profit{\J}{r} > 0\}
    \iftoggle{widecol}{}{\\ &}
    = \inf\{r \geq 0 \mid \sup_\pi (rP_\pi - E_\pi) > 0\}
  \iftoggle{widecol}{\end{equation*}}{\end{align*}}
  and
  \begin{equation*}
    \frac{1}{\gittins{\J}}
    = \inf_\pi \frac{E_\pi}{P_\pi}
    = \inf\biggl\{r \geq 0 \biggm|
      r > \inf_\pi \frac{E_\pi}{P_\pi}\biggr\}.
  \end{equation*}
  Each of $\sup_\pi (rP_\pi - E_\pi) > 0$
  and $r > \inf_\pi E_\pi/P_\pi$
  holds if and only if there exists a stopping policy~$\pi$
  such that $rP_\pi > E_\pi$,
  so $\fair{\J}$ and $1/\gittins{\J}$ are infima of the same set.
\end{proof}

By \fref{prop:gittins_fair},
scheduling the job of greatest Gittins index is equivalent to
scheduling the job of least SJP index.
Hereafter, we work mostly in terms of SJP functions and SJP indices.

\subsection{Prior Gittins Index Definitions}
\label{sub:prior_gittins}

The Gittins index was originally studied in the context of
the multi-armed bandit problem,
for which it was developed by the eponymous Gittins
\citep{original_gittins, review_gittins},
but the same technique has since been applied to a host of other
stochastic optimization problems.
See \citet{book_gittins} for a recent overview.
The majority of works on the Gittins index
consider it in a \emph{discounted} setting,
where cost incurred or reward received at time~$t$
is scaled by a factor of $e^{-\alpha t}$
for some discount rate $\alpha \geq 0$.

However, minimizing mean response time in a queueing system
is a problem with \emph{undiscounted} costs.
Works that consider the Gittins index in an undiscounted setting
take one of two approaches.
The first approach is
to introduce discounting and consider the $\alpha \to 0$ limit.
\Citet{tax_undiscounted_whittle} takes this approach to
analyze the mean response time of the Gittins policy
in a model where each stage of each job has exponential size distribution.
The second approach is to avoid discounting altogether.
\Citet{golf_dumitriu} take this approach to solve a discrete-time problem
which, roughly speaking,
can be thought of as minimizing the completion time of
the first job to exit the system.

Broadly speaking, we take the approach of avoiding discounting altogether.
SJP is the continuous analogue
of the ``token vs. terminator game'' used by \citet{golf_dumitriu},
which in turn was an undiscounted analogue of a similar construction
introduced by \citet{gittins_index_whittle}.
Our contributions are
connecting SJP to
the traditionally defined Gittins index (\fref{prop:gittins_fair}),
showing that SJP admits a natural composition property
(\fref{thm:composition}), and
using SJP to analyze mean response time (\fref{thm:queueing_time}).
The resulting mean response time formula
is similar to that of \citet{tax_undiscounted_whittle}
but generalized to handle arbitrary stage size distributions.

\section{Single-Job Profit Composition Law}
\label{sec:composition}

The main benefit of SJP over other Gittins index formulations
is that it admits an elegant \emph{composition law} (\fref{thm:composition}).
The composition law says that
if a multistage job can written as a \emph{sequential composition}
of two jobs $\J$ and~$\K$ (\fref{def:sequential}),
then the SJP function of the job can be written in terms of
$\profit{\J}{}$ and~$\profit{\K}{}$.

\begin{definition}
  \label{def:sequential}
  The \emph{sequential composition} of jobs $\J$ and~$\K$,
  denoted $\J \seq \K$,
  is the job obtained by ``stitching together'' $\J$ and $\K$:
  the initial stage is $\J$'s initial stage,
  and all transitions to $\J$'s final stage go to $\K$'s initial stage instead.
  See \fref{fig:sequential_composition} for an illustration.

  \begin{figure}
    \centering
    \begin{tikzpicture}[scale=0.8, thick]
      \tikzstyle{EdgeStyle}=[->, >=stealth']
      \node at (-3, 1.1) {$\J$};
      \node[draw, rounded corners=3pt, inner sep=8]
        at (-3, 0) {\small$\J$};
      \node at (0, 1.1) {$\K$};
      \node[draw, diamond, rounded corners=3pt, inner sep=5]
        at (0, 0) {\small$\K$};
      \node at (3, 1.1) {$J \seq \K$};
      \node[draw, rounded corners=3pt, inner sep=8] (J)
        at (3, 0) {\small$\J$};
      \node[draw, diamond, rounded corners=3pt, inner sep=5] (K)
        at (3, -1.6) {\small$\K$};
      \path (K.north) +(0, -2pt) coordinate (Kn);
      \Edges(J.south,Kn);
    \end{tikzpicture}

    \nocaptiondetail

    \caption{Sequential Composition (\fref{def:sequential})}
    \label{fig:sequential_composition}
  \end{figure}
\end{definition}

\Fref{def:sequential} is the first of several ways of constructing jobs,
summarized in \fref{tab:construction_notation},
which we define in this section.
We give a more formal version of it in \fref{app:definitions}.

\begin{table}
  \caption{Summary of Job Construction Notation}
  \label{tab:construction_notation}

  \centering
  \squish\begin{tabular}{@{}lll@{}}
    \toprule
    \emph{Notation} & \emph{Description} & \emph{Defined in} \\
    \midrule
    $\J \seq \K$
      & sequential composition
      & \fref{def:sequential} \\
    $\single{i}$
      & single-stage job
      & \fref{def:single-stage} \\
    $\{q_\ell : \J_\ell\}_{\ell \in \mc{L}}$ or $\J_{\mc{L}}$
      & mixture composition
      & \fref{def:mixture} \\
    \bottomrule
  \end{tabular}\squish
\end{table}

Remarkably, we can easily express the SJP function of a sequential composition
in terms of the SJP functions of its components.

\begin{theorem}[Composition Law]
  \label{thm:composition}
  For any jobs $\J$ and~$\K$ and any reward~$r$,
  \begin{equation*}
    \profit{\J \seq \K}{r} = \profit{\J}{\profit{\K}{r}}.
  \end{equation*}
  That is, $\profit{\J \seq \K}{} = \profit{\J}{} \circ \profit{\K}{}\esub$.
\end{theorem}

\begin{proof}
  A stopping policy~$\pi_{\J \seq \K}$ for job~$\J \seq \K$ is a vector
  containing a stopping age $\pi_{\J \seq \K}(i)$ for each stage
  $i \in \stages[\J \seq \K]$.
  Each stage~$i$ is in exactly one of $\stages[\J]$ and $\stages[\K]$,
  so we can think of $\pi_{\J \seq \K}$ as a pair vectors:
  one a stopping policy~$\pi_\J$ for~$\J$,
  and the other a stopping policy~$\pi_\K$ for~$\K$.
  Let
  \begin{align*}
    P_\J &= \P{\text{$\pi_\J$ completes~$\J$}} \\
    E_\J &= \E{\text{time $\pi_\J$ serves~$\J$}}
  \end{align*}
  and similarly for $\K$ and~$\J \seq \K$.
  Thinking of $\J \seq \K$ as $\J$ ``followed by''~$\K$, we have
  $P_{\J \seq \K} = P_\J P_\K$ and $E_{\J \seq \K} = E_\J + P_\J E_\K$,
  from which we compute
  \begin{align*}
    \profit{\J \seq \K}{r}
    &= \sup_{\mathclap{\pi_{\J \seq \K}}} (rP_{\J \seq \K} - E_{\J \seq \K}) \\
    &= \sup_{\mathclap{\pi_\J, \pi_\K}} (rP_\J P_\K - (E_\J + P_\J E_\K)) \\
    &= \sup_{\pi_\J} (\sup_{\pi_\K}(rP_\K - E_\K)P_\J - E_\J) \\
    &= \profit{\J}{\profit{\K}{r}}.
    \qedhere
  \end{align*}
\end{proof}

\Fref{thm:composition} shows that SJP functions naturally compose.
We can also use it to see that SJP indices or Gittins indices alone
do \emph{not} naturally compose.
To see this, observe that
\begin{equation*}
  \fair{\J \seq \K}
  = \profitinv{\J \seq \K}{0}
  = \profitinv{\K}{\profitinv{\J}{0}}
  = \profitinv{\K}{\fair{\J}}.
\end{equation*}
Consider the simple case where
$\J$ is a single-stage job of deterministic size~$d$.
Then $\fair{\J \seq \K} = \profitinv{\K}{d}$,
so to handle arbitrary~$d$,
we need to know $\K$'s entire SJP \emph{function}, not just its SJP index.

\subsection{Reduction to Single-Stage Jobs}
\label{sub:single-stage}

It may seem at first that \fref{thm:composition} has a limited scope,
because it only applies to jobs that are sequential compositions.
Fortunately, it is possible to decompose \emph{any} multistage job~$\J$ into
a sequential composition of two components,
one of which is a single-stage job with size distribution~$\size{\init[\J]}$.
By applying this decomposition recursively,
we can obtain the SJP function of any multistage job
in terms of single-stage SJP functions.
To express the decomposition, we need the following definitions,
which we formalize in \fref{app:definitions}.

\begin{definition}
  \label{def:single-stage}
  The \emph{single-stage job} with stage~$i$, denoted~$\single{i}$,
  is the job with just one stage~$i$ of size~$\size{i}$.
\end{definition}

\begin{definition}
  \label{def:mixture}
  Let $\mc{L}$ be a finite set.
  The \emph{mixture composition} of jobs $\{\J_\ell\}_{\ell \in \mc{L}}$
  with probabilities $\{q_\ell\}_{\ell \in \mc{L}}$
  is the job representing a randomly chosen job,
  choosing $\J_\ell$ with probability~$q_\ell$.
  See \fref{fig:mixture_composition} for an illustration.
  We denote the mixture composition by
  $\{\J_\ell \text{ w.p. } q_\ell\}_{\ell \in \mc{L}}$,
  abbreviating to $\J_{\mc{L}}$
  when the probabilities are unambiguous.

  \begin{figure}
    \centering
    \begin{tikzpicture}[scale=0.8, thick]
      \SetUpVertex[Math,LabelOut]
      \tikzstyle{EdgeStyle}=[->, >=stealth']
      \tikzstyle{VertexStyle}=[shape=circle, minimum width=1em, draw]
      \newcommand{\leftlbl}{left, fill=none, shift={(0,0.1)}}
      \newcommand{\rightlbl}{right, fill=none, shift={(0,0.1)}}
      \node at (-3, 1.1) {$\J$};
      \node[draw, rounded corners=3pt, inner sep=8]
        at (-3, 0) {\small$\J$};
      \node at (0, 1.1) {$\K$};
      \node[draw, diamond, rounded corners=3pt, inner sep=5]
        at (0, 0) {\small$\K$};
      \node at (3, 1.1) {$\{p : J, 1 - p : \K\}$};
      \Vertex[x=3, y=0.4, L={0}]{A}
      \node[draw, rounded corners=3pt, inner sep=8] (J)
        at (2.25, -1) {\small$\J$};
      \node[draw, diamond, rounded corners=3pt, inner sep=5] (K)
        at (3.75, -1) {\small$\K$};
      \path (K.north) +(0, -2pt) coordinate (Kn);
      \Edge[label=$p$, labelstyle={\leftlbl}, style={pos=0.5425}](A)(J.north);
      \Edge[label=$1 - p$, labelstyle={\rightlbl}, style={pos=0.6}](A)(Kn);
    \end{tikzpicture}

    \nocaptiondetail

    \caption{Mixture Composition (\fref{def:mixture})}
    \label{fig:mixture_composition}
  \end{figure}
\end{definition}

We can express any multistage job as a sequential composition
of a single-stage job,
namely its first stage,
and a mixture composition,
namely a mixture of the possible state-conditioned jobs that might arise
after the first stage transition.
We can thus write
\begin{equation}
  \label{eq:decomposition}
  \J
  = \single{\init[\J]} \seq \{\st[\J]{i}{0}
      \text{ w.p. } \trans[\J]{\init[\J]}{i}\}_{i \in \stages[\J]}.
\end{equation}
It is simple to check that the optimal profit of a mixture composition
is the expected optimal profit of the randomly chosen job:
\begin{equation}
\label{eq:profit_mixture}
  \profit{\J_{\mc{L}}}{r}
  = \sum_{\mathclap{\ell \in \mc{L}}} q_\ell\profit{\J_\ell}{r}.
\end{equation}
Combining \fref{eq:decomposition} and \fref{eq:profit_mixture}
with \fref{thm:composition}
and adopting the convention $\profit{\st[\J]{\fin[\J]}{0}}{r} = r$
yields the following.

\begin{corollary}
  \label{cor:decomposition}
  For any job~$\J$ and any reward~$r$,
  \begin{equation*}
    \profit{\J}{r}
    = \profit{\single{\init[\J]}}{}\Bigl(
        \sum_{\mathclap{i \in \stages[\J]}}
          \trans[\J]{\init[\J]}{i} \profit{\st[\J]{i}{0}}{r}
      \Bigr).
  \end{equation*}
\end{corollary}

By recursively applying \fref{cor:decomposition},
we can express the SJP function of a multistage job
in terms of single-stage SJP functions,
as shown in \fref{alg:gittins}.
This reduces SJP for a multistage job,
which is a multidimensional optimization problem,
to several single-stage SJP instances,
each of which is a single-dimensional optimization problem.
Thanks to the convenient properties proved in \fref{lem:sjp_properties},
a simple numerical technique such as bisection
suffices to find a job's SJP index given its SJP function.

\begin{algorithm}
  \caption{Gittins Index of a Multistage Job}
  \label{alg:gittins}
  \justifying\noindent
  \textsc{Input:} Multistage job~$\J$. \\
  \textsc{Output:} Gittins index~$\gittins{\J}$.
  \begin{enumerate}[(i)]
  \item
    Define $\profit{\st[\J]{\fin[\J]}{0}}{r} = r$.
  \item
    For every $i \in \stages[\J] \setminus \{\fin[\J]\}$,
    define $\profit{\st[\J]{i}{0}}{}$
    using \fref{cor:decomposition}:
    \begin{equation*}
      \profit{\st[\J]{i}{0}}{r}
      = \profit{\single{i}}{}\Bigl(
          \sum_{\mathclap{i \in \stages[\J]}}
            \trans[\J]{\init[\J]}{i} \profit{\st[\J]{i}{0}}{r}
        \Bigr).
    \end{equation*}
    These functions are well-defined because stage transitions are acyclic,
    so there is no self-reference in the non-zero terms.
    The SJP function of~$\J$ is
    $\profit{\J}{} = \profit{\st[\J]{\init[\J]}{0}}{}$.
  \item
    Use bisection or similar
    to compute the largest~$r$ such that $\profit{\J}{r} = 0$.
    This value is $\fair{\J} = 1/\gittins{\J}$.
  \end{enumerate}
\end{algorithm}

\Fref{alg:gittins} tells us that
to compute the SJP index of a multistage job,
it suffices to compute the SJP function of several single-stage jobs.
By \fref{def:fair},
to compute the SJP function of a single-stage job,
it suffices to find, for each reward~$r$,
the earliest age~$x(r)$ at which
the SJP index first exceeds~$r$.
We can obtain~$x(r)$
from prior work on the Gittins index in the M/G/1 setting
\citep{m/g/1_gittins_aalto, mlps_gittins_aalto}.

\subsection{Discrete Stage Size Distributions}
\label{sub:discrete}

It is natural to ask what the computational complexity of \fref{alg:gittins} is
and how that complexity compares to prior approaches.
When stages have continuous distributions this is difficult to answer,
because to the best of our knowledge
the complexity of computing the Gittins index
for even a single-stage job with continuous size distribution
has not been previously studied.
Therefore, to demonstrate the impact of the composition law
on computational complexity,
we consider jobs with discrete stage size distributions.

Let us consider the special case of a multistage job~$\J$
with stage size distributions of finite support
and a deterministic stage sequence,
meaning we can write $\J = \single{i_1} \seq \dotsb \seq \single{i_m}$.
We can use \fref{thm:composition} to obtain the Gittins index of~$\J$
in $O(n \log n)$ time.
Here $n = |\J| = \sum_{i \in \stages[J]} |\size{i}|$
is the size of $\J$'s state space,
where $|\size{i}|$ is the number of support points of stage~$i$.

The algorithm first computes the SJP function~$\profit{\single{i}}{}$
for each stage~$i$, then does the following:
\begin{enumerate}[(i)]
\item
  Find stage~$i$ such that $\J = \K \seq \single{i} \seq \L$
  with $|\K|, |\L| \leq n/2$.\footnote{%
    If the first or last stage has many support points,
    then $i$ may be the first or last stage,
    in which case we omit $\K$ or $\L$, respectively.
    Indeed, in the base case, $i$ is the only stage.}
\item
  Recursively compute $\profit{\K}{}$ and $\profit{\L}{}$.
\item
  Return
  $\profit{\J}{}
    = \profit{\K}{} \circ \profit{\single{i}}{} \circ \profit{\L}{}$.
\end{enumerate}
We show the following two facts in \fref{app:discrete}:
\begin{itemize}
\item
  We can compute $\profit{\single{i}}{}$ in $O(|X_i|)$ time.
\item
  We can compose
  $\profit{\K}{} \circ \profit{\single{i}}{} \circ \profit{\L}{}$
  in $O(|\K| + |X_i| + |\L|)$ time.
  This is because the SJP functions are piecewise linear,
  so we can represent them as linked lists such that
  function composition works much like the ``merge'' step of merge sort.
\end{itemize}
Therefore, the algorithm above takes $O(n \log n)$ time.

Our $O(n \log n)$ algorithm improves upon
the best algorithms in the literature, which take $O(n^3)$ time
\citep{gittins_index_computation_chakravorty}.
However, previous algorithms consider a more general problem
with arbitrary stage transitions.

\section{Response Time Analysis}
\label{sec:analysis}

In this section, we analyze the response time of the Gittins policy.
We actually focus on queueing time,
which is response time minus service time.
Recall from \fref{sub:conventions} that we assume all jobs have
the same job type~$\J$,
so we omit the $\J$ from the notation in \fref{tab:job_notation}.
In particular,
we write $\profit{\st{i}{x}}{r} = \profit{\st[\J]{i}{x}}{r}$
and $\totalsize{\st{i}{x}} = \totalsize{\st[\J]{i}{x}}$.

Our approach is very different from the tagged-job approach
that is often used to analyze response time
\citep{book_harchol-balter, book_kleinrock, soap_scully,
  fb_analysis_schrage, srpt_analysis_schrage}.
Instead, we analyze response time
by solving a \emph{continuous dynamic program}.
We find a \emph{cost function} that,
for every possible state of the system,
gives the \emph{expected total queueing time}
for the remainder of the busy period.
This total is the sum of queueing times
of all jobs that appear during the busy period,
including those currently in the system and those that arrive later.
From this cost function, mean queueing time follows easily.

We are not the first to use a dynamic programming approach
to analyze mean response time:
\citet{tax_undiscounted_whittle}
analyzes multistage jobs whose states have exponential size distributions,
and \citet{rs_slowdown_hyytia}
analyze jobs with known sizes.
We apply dynamic programming to a system
in which stages have \emph{unknown size} and \emph{general size distributions},
which requires new techniques.

Our main result is the following expression for
queueing time under the Gittins policy,
which is in terms of the SJP function and its derivative.

\begin{theorem}[Gittins Policy Queueing Time]
  \label{thm:queueing_time}
  The mean queueing time of multistage jobs
  under the Gittins policy is
  \begin{equation*}
    \E{\Tqueue}
    = \sum_{\mathclap{i \in \stages}}\mkern5mu
        \integral{0}{\infty}\mkern-1.5mu
          \integral{0}{\infty}
            \lambda q_{ai} \tail[i]{x}
            \frac{%
              \dd{r}\costB{r}{\st{i}{x}}\nudgel
              \dd{r}\costB{r}{\st{a}{0}}}{%
              \dd{r}\frozen{r}}
        \,dr
      \,dx,
  \end{equation*}
  where $q_{ai}\esub$ is the probability that
  a job at some point reaches stage~$i$ and
  \begin{align*}
    \load{r} &= \lambda(r - \profit{\st{a}{0}}{r}) \\
    \frozen{r} &= \frac{r}{1 - \load{r}} \\
    \costB{r}{\st{i}{x}} &= \frac{r - \profit{\st{i}{x}}{r}}{1 - \load{r}}.
  \end{align*}
\end{theorem}

To prove \fref{thm:queueing_time}, we
\begin{itemize}
\item
  write an optimality equation
  that characterizes the cost function of the dynamic program
  (\fref{sub:optimality});
\item
  as a warmup, solve a simplified dynamic program
  without arrivals (\fref{sub:warmup}); and
\item
  using intuition from the arrival-free case,
  guess a cost function for the full dynamic program with arrivals
  and verify it satisfies the optimality equation
  (\fref{sub:guess}).
\end{itemize}
It turns out that to understand the intuition behind
our guess for the cost function,
it helps to generalize the scheduling problem we consider (\fref{sub:bpop}).
We consider a setting where
there are not one but two possible actions available for each job:
\emph{serve} and \emph{bypass}.
Serving a job advances its state as usual,
and bypassing gives the job an alternate way to complete.
Specifically, while a job is being bypassed,
it stochastically jumps to state $\st{\fin}{0}$ at rate $1/r$,
where $r \geq 0$ is a parameter representing the \emph{expected bypass time}.

Scheduling with bypassing is simple to characterize for extreme values of~$r$.
\begin{itemize}
\item
  In the $r \to 0$ limit,
  we bypass all jobs and the total queueing time approaches~$0$.
\item
  In the $r \to \infty$ limit, we never bypass any jobs,
  so we recover ordinary scheduling without bypassing.
\end{itemize}
Ultimately, we want to understand total queueing time
in the $r \to \infty$ limit.
Given that we understand the $r \to 0$ limit,
it would suffice to compute the \emph{derivative with respect to~$r\esub$}
of total queueing time,
then integrate.
This is exactly what we do in \fref{sub:warmup} for the arrival-free case,
and a similar approach works when incorporating arrivals in \fref{sub:guess}.

The option to bypass a job is intimately related to giving up in SJP,
as hinted at by the common notation~$r$.
Specifically, SJP with reward~$r$ is equivalent to
minimizing the completion time of a single job with expected bypass time~$r$.
This is because we can think of $r$ as the \emph{opportunity cost}
of not completing the job in SJP.
Of course, the relationship between bypassing and SJP is less clear
in a context where new jobs might arrive during a bypass.
Finding this relationship
is one of the obstacles we overcome in \fref{sub:guess}
when adapting our arrival-free cost function from \fref{sub:warmup}
to the setting with arrivals.

\subsection{Busy Period Optimization Problems}
\label{sub:bpop}

In this section we formulate several optimization problems
as continuous dynamic programs.
All such problems involve minimizing a cost
over the course of a busy period,
so we call them \emph{busy period optimization problems} (BPOPs).
We will consider five BPOPs in total,
called \emph{BPOP~A}, \emph{BPOP~B}, \emph{BPOP~C}, \emph{BPOP~P},
and \emph{BPOP~Q}.

BPOPs P and~Q are \emph{total queueing time minimization} problems.
At any moment in time, there may be several jobs in the system,
and we must choose which job to serve.
Our goal is to minimize the expected total queueing time
of all jobs in the busy period.
These two BPOPs differ only in the arrival process.
\begin{itemize}
\item
  In BPOP~P, no arrivals occur.
\item
  In BPOP~Q, arrivals occur at rate~$\lambda$.
\end{itemize}
BPOP~Q is the problem we ultimately want to solve,
towards which BPOP~P is a helpful intermediate step.

BPOPs A, B, and~C are all \emph{busy period length minimization} problems.
While busy period length may not seem directly related to queueing time,
these three BPOPs are useful for solving BPOPs P and~Q.
At any moment in time in BPOPs A~, B, and~C,
we can choose to \emph{serve} any job or \emph{bypass} any job.
Our goal is to minimize the expected amount of time until the system is empty.
These three BPOPs differ only in the arrival process.
\begin{itemize}
\item
  In BPOP~A, no arrivals occur.
\item
  In BPOP~B, arrivals occur at rate~$\lambda$.
\item
  In BPOP~C, arrivals occur at rate~$\lambda$ while serving a job,
  but no arrivals occur while bypassing a job.
\end{itemize}
Note that BPOP~A with a single job is essentially equivalent to SJP.

To formulate a BPOP as a dynamic program, we must define its
state space, action space, state dynamics, and cost dynamics.

The \emph{state space} of a BPOP is the set of \emph{lists of job states}.
That is, when there are $n$ jobs in the system,
the \emph{system state}, or simply \emph{state} when not ambiguous, is a list
\begin{align*}
  \sts &= (\st{i_1}{x_1}, \dots, \st{i_n}{x_n}),
\end{align*}
where $\st{i_k}{x_k}$ is the state of job~$k$.
To keep the notation focused on the important components of system states,
we use ellipses to hide the unimportant parts.

The \emph{action space} of a BPOP
consists of up to two possible actions for each job~$k$:
\emph{bypass}, which is possible in BPOPs A, B, and~C;
and \emph{serve}, which is possible in all BPOPs.

The \emph{state dynamics} of a BPOP depend on the action taken
and vary according to details of the BPOP in question,
but they all have the following aspects in common:
\begin{itemize}
\item
  When bypassing or serving job~$k$,
  all other jobs remain in the same state.
\item
  When bypassing job~$k$,
  its state stochastically jumps from $\st{i_k}{x_k}$ to $\st{\fin}{0}$
  at rate~$1/r$,
  where $r \geq 0$ is a constant.
\item
  When serving job~$k$,
  its state evolves as described in \fref{sec:model}:
  its age~$x_k$ increases at rate~$1$,
  and its state jumps stochastically from $\st{i_k}{x_k}$ to $\st{j}{0}$
  at rate~$\trans{i_k}{j}\hazard{i_k}{x_k}$.
\item
  When arrivals occur,
  the system state jumps from $(\dots)$ to $(\dots, \st{a}{0})$
  at stochastic rate $\lambda$.
\end{itemize}
For example, while serving job~$k$ in BPOP~B, C, or~Q,
the system state evolves accounting for both $k$'s service and new arrivals.

The \emph{cost dynamics} of a BPOP are simple:
we pay cost continuously at state-dependent rate.
The cost rate is always~$1$ for BPOPs A, B, and~C;
and the cost rate is $n-1$, the number of jobs in the queue,
for BPOPs P and~Q.

The \emph{cost function} of a BPOP maps each system state~$\sts$
to the expected cost that will be paid \emph{under an optimal policy}
as the system advances from $\sts$ to the empty state~$()$.
We write $\costA{r}{}$ for the cost function of BPOP~A,
and similarly for $\costB{r}{}$, $\costC{r}{}$, $\costP{}$, and~$\costQ{}$.
Note that the cost functions for BPOPs A, B, and~C
are parametrized by the expected bypass time~$r$.
To refer to a generic cost function mapping system states to costs,
we write~$\cost{}$.
Every cost function~$\cost{}$ discussed in this section
satisfies the following conditions:
\begin{itemize}
\item
  The order of job states is arbitrary,
  so for all $1 \leq k < \ell \leq n$,
  \begin{equation*}
    \cost{\dots, \st{i_k}{x_k}, \dots, \st{i_\ell}{x_\ell}, \dots}
    = \cost{\dots, \st{i_\ell}{x_\ell}, \dots, \st{i_k}{x_k}, \dots}.
  \end{equation*}
\item
  Jobs depart upon entering stage~$\fin$,
  so $\cost{\dots, \st{\fin}{0}} = \cost{\dots}$.
\item
  We measure cost until the busy period ends,
  which is when the system is empty, so
  $\cost{{}} = 0$.
\end{itemize}
These serve as boundary conditions when characterizing a BPOP's cost function.

\subsection{Optimality Equations}
\label{sub:optimality}

We now turn our attention to characterizing
the cost functions of the five BPOPs.
Our ultimate goal in proving \fref{thm:queueing_time}
is to obtain an expression for mean queueing time~$\E{\Tqueue}$.
The BPOP most directly related to queueing time is BPOP~Q:
a simple argument using renewal theory yields
\begin{equation}
  \label{eq:renewal_reward}
  \E{\Tqueue} = (1 - \load{})\costQ{\st{a}{0}},
\end{equation}
where $\load{}$ is the system load
and $\costQ{}$ is the cost function of BPOP~Q.
This is because each busy period has
total expected queueing time $\costQ{\st{a}{0}}$
spread over $1/(1 - \load{})$ jobs in expectation.
Of course, to make use of this observation,
we must compute $\costQ{\st{a}{0}}$.
The statement of \fref{thm:queueing_time} previews the final answer,
which we see involves $\costB{r}{}$, the cost function of BPOP~B,
and the other BPOPs are helpful during the derivation.

We have given the
state space, action space, state dynamics, and cost dynamics
of BPOP~Q in \fref{sub:bpop}.
We can use this to formulate BPOP~Q as a dynamic program
using standard techniques,
specifically those for piecewise-deterministic Markov processes
\citep{piecewise-deterministic_davis, piecewise-deterministic_vermes}.
From this formulation
one can obtain an \emph{optimality equation},
specifically a Hamilton-Jacobi-Bellman equation
\citep{piecewise-deterministic_vermes},
which characterizes~$\costQ{\sts}$.
For brevity,
we present the optimality equation for BPOP~Q and the other BPOPs directly.

The optimality equation for any BPOP
says that for every system state and admissible action in that state,
the following quantities should add to at least~$0$:
\begin{enumerate}[(i)]
\item
  \label{item:cost_incurred}
  the expected rate cost is incurred according to the cost dynamics and
\item
  \label{item:cost_function_change}
  the expected rate at which the cost function changes
  as the state evolves according to the state dynamics.
\end{enumerate}
Furthermore, the quantities must add to \emph{exactly}~$0$
when the \emph{optimal} action is taken.
Quantity~\fref{item:cost_incurred}
is simply $1$ or $n - 1$, depending on the BPOP.
We express quantity~\fref{item:cost_function_change}
in terms of three operators on cost functions,
which give the rate at which a cost function changes
due to various components of the state dynamics.

\begin{definition}
  \label{def:operators}
  The \emph{serve-$k$ operator}, denoted~$\serve{k}$,
  gives the expected rate at which a cost function~$\cost{}$ changes
  due to serving job~$k$:
  \begin{align*}
      \iftoggle{widecol}{}{&}
      \serve{k} \cost{\dots, \st{i_k}{x_k}, \dots}
      \iftoggle{widecol}{&}{\\ &\qquad}
        = \dd{x_k} \cost{\dots, \st{i_k}{x_k}, \dots} \\
      &\iftoggle{widecol}{\quad}{\qqquad}
        + \hazard{i_k}{x_k}\Bigl(
        \sum_{\mathclap{j \in \stages}}
          \trans{i}{j}\cost{\dots, \st{j}{0}, \dots}
        - \cost{\dots, \st{i_k}{x_k}, \dots}\Bigr).
  \end{align*}
  The \emph{arrive operator}, denoted~$\arrive$,
  gives the expected rate at which a cost function~$\cost{}$ changes
  due to new jobs arriving:
  \begin{equation*}
    \arrive \cost{\dots} = \lambda(\cost{\dots, \st{a}{0}} - \cost{\dots}).
  \end{equation*}
  The \emph{bypass-$k$ operator}, denoted~$\bypass{r}{k}$,
  gives the expected rate at which a cost function~$\cost{}$ changes
  due to bypassing job~$k$:
  \iftoggle{widecol}{\begin{equation*}}{\begin{align*}}
    \bypass{r}{k} \cost{\dots, \st{i_k}{x_k}, \dots}
    \iftoggle{widecol}{}{&}
    = \frac{1}{r}(\cost{\dots, \st{\fin}{0}, \dots}
    \iftoggle{widecol}{}{\\ &\qquad}
      - \cost{\dots, \st{i_k}{x_k}, \dots}).
  \iftoggle{widecol}{\end{equation*}}{\end{align*}}
\end{definition}

To be concrete, consider BPOP~P in state~$\sts$ with $n \geq 1$ jobs.
At every moment in time, we incur cost at rate $n - 1$.
The rate at which the cost function changes is given by
the serve-$k$ operator~$\serve{k}$,
depending on which job~$k$ we serve.
No matter which job we serve,
the cost function should change at expected rate
greater than or equal to $-(n - 1)$,
achieving equality when serving the optimal job.
Thus, the optimality equation for BPOP~P is
\begin{equation}
  \label{eq:opt_p}
  \refstepcounter{equation}
  \tag{\theequation/P*}
  \min_k \serve{k}\costP{\sts} = -(n - 1).
\end{equation}
Above, and in all future optimality equations,
the minimum is understood as taken over $1 \leq k \leq n$.
The optimality equation for BPOP~Q is similar,
except now arrivals contribute to the rate at which
the cost function changes:
\begin{equation}
  \label{eq:opt_q}
  \refstepcounter{equation}
  \tag{\theequation/Q*}
  \min_k{}(\serve{k} + \arrive)\costQ{\sts} = -(n - 1).
\end{equation}

The optimality equations for BPOPs A, B, and~C
are similar to those for BPOPs P and~Q with two main differences.
First, they incur cost at rate~$1$, so the right-hand side is $-1$.
Second, we have the option to bypass jobs instead of serving them,
manifesting as a nested minimization,
first deciding which job to act on and then deciding
whether to serve or bypass it.
For readability, we write the inner minimization vertically.
\begin{itemize}
\item
  BPOP~A has no arrivals, so its optimality equation is
  \begin{equation}
    \label{eq:opt_a}
    \refstepcounter{equation}
    \tag{\theequation/A*}
    \min_k \min\biggl\{
    \begin{matrix}
      \serve{k}\costA{r}{\sts} \\
      \bypass{r}{k}\costA{r}{\sts}
    \end{matrix}
    \biggr\} = -1.
  \end{equation}
\item
  BPOP~B always has arrivals, so its optimality equation is
  \begin{equation}
    \label{eq:opt_b}
    \refstepcounter{equation}
    \tag{\theequation/B*}
    \min_k \min\biggl\{
    \begin{matrix}
      (\serve{k} + \arrive)\costB{r}{\sts} \\
      (\bypass{r}{k} + \arrive)\costB{r}{\sts}
    \end{matrix}
    \biggr\} = -1.
  \end{equation}
\item
  BPOP~C has arrivals while serving but not while bypassing,
  so its optimality equation is
  \begin{equation}
    \label{eq:opt_c}
    \refstepcounter{equation}
    \tag{\theequation/C*}
    \min_k \min\biggl\{
    \begin{matrix}
      (\serve{k} + \arrive)\costC{r}{\sts} \\
      \bypass{r}{k}\costC{r}{\sts}
    \end{matrix}
    \biggr\} = -1.
  \end{equation}
\end{itemize}

There are a few details to clarify about the optimality equations above.
We have implicitly assumed that
the system state~$\sts$ has $n \geq 1$ jobs,
none of which are in the final stage~$z$.
If a job is in stage~$z$ or the system has no jobs,
a boundary condition applies as described in \fref{sub:bpop}.
Finally, the optimality equations admit spurious solutions
which we rule out in \fref{app:spurious}.

\subsection{Guessing the Cost Function: No Arrivals}
\label{sub:warmup}

As a warmup for guessing the cost function for BPOP~Q,
we focus on BPOP~P,
the problem of minimizing total expected queueing time
in an arrival-free setting.
The first step to understanding BPOP~P
is understanding the even simpler BPOP~A,
in which we want to minimize busy period length.

\begin{lemma}
  \label{lem:cost_a}
  The cost function of BPOP~A is
  \begin{equation*}
    \costA{r}{\sts}
    = \sum_k \costA{r}{\st{i_k}{x_k}}
    = \sum_k (r - \profit{\st{i_k}{x_k}}{r}).
  \end{equation*}
\end{lemma}

\begin{proof}
  We will prove the desired formula for states with a single job.
  To extend to the case of multiple jobs,
  observe that the busy period length is the total time spent on all jobs,
  so we can minimize the time spent on each job individually.
  To show $\costA{r}{\st{i}{x}} = r - \profit{\st{i}{x}}{r}$
  consider the following variants of SJP with job~$\st{i}{x}$ and reward~$r$,
  which are clearly equivalent:
  \begin{enumerate}[(i)]
  \item
    \label{item:sjp}
    maximizing expected reward earned minus time spent,
    with reward~$r$ for completion
    and reward~$0$ for giving up; and
  \item
    \label{item:bpop_a}
    maximizing expected reward earned minus time spent,
    with reward~$0$ for completion
    and reward~$-r$ for giving up.
  \end{enumerate}
  Problem~\fref{item:sjp} is ordinary SJP,
  while problem~\fref{item:bpop_a} is equivalent to BPOP~A with a single job:
  there is no reward for completion,
  but we can give up by bypassing the job.
  Here we use the fact that there exists a Markovian optimal policy
  justify the equivalence of bypassing in expected time~$r$
  and giving up for reward~$-r$ are equivalent,
  because a Markovian policy will never switch from bypassing back to serving.
  The optimal profit in problem~\fref{item:bpop_a}
  is $\profit{\st{i}{x}}{r} - r$,
  and cost is negative profit,
  from which the desired formula follows.
\end{proof}

Relating the cost function of BPOP~A to SJP functions
lets us apply \fref{lem:sjp_properties}.
Specifically,
we learn that $\dd{r}\costA{r}{\st{i}{x}}$, as a function of~$r$,
is nonincreasing and bounded between $0$ and~$1$,
meaning it is a valid tail function of some random variable.

\begin{definition}
  \label{def:prevailing}
  The \emph{prevailing index} of a job in state $\st{i}{x}$,
  denoted~$\prevailing{i}{x}$,
  is a random variable with tail function
  \begin{equation*}
    \P{\prevailing{i}{x} > r} = \dd{r}\costA{r}{\st{i}{x}}.
  \end{equation*}
  For our purposes,
  the prevailing index $\prevailing{i}{x}$ is a computational tool
  used for its convenient distribution.
  When two prevailing indices appear in the same formula,
  we assume they are independent.
\end{definition}

There are two important properties of the prevailing index.
\Fref{def:fair} implies
\begin{equation}
  \label{eq:prevailing_geq_fair}
  \P{\prevailing{i}{x} \geq \fair{\st{i}{x}}} = 1,
\end{equation}
and integrating the tail function of the prevailing index yields
\begin{equation}
  \label{eq:cost_a_prevailing}
  \E{\min\{\prevailing{i}{x}, r\}}
  = \integral{0}{r} \P{\prevailing{i}{x} > s} \,ds
  = \costA{r}{\st{i}{x}}.
\end{equation}
Taking the $r \to \infty$ limit of~\fref{eq:cost_a_prevailing},
we find $\E{\prevailing{i}{x}} = \E{\totalsize{\st{i}{x}}}$
by Lemmas~\ref{lem:sjp_properties} and~\ref{lem:cost_a}.
That is, a state's prevailing index and remaining total size
have the same expectation.

We now return our focus to BPOP~P.
We already know the optimal policy for BPOP~P,
namely always serving the job of minimal SJP index,
but it remains to find its cost function~$\costP{}$.
Because there are no arriving jobs,
we guess that expected total queueing time has the form
\begin{equation*}
  \costP{\sts}
  = \sum_{\mathclap{k < \ell}} \costP{\st{i_k}{x_k}, \st{i_\ell}{x_\ell}}.
\end{equation*}
To see why this holds,
note that we can think of queueing time as
the sum over all pairs of jobs of
the ``interference'' between the two jobs.
By interference,
we mean the amount of time such that one of the jobs is waiting while
the other job is in service.
Because we always serve the job of minimal SJP index,
the interference between two jobs $k$ and~$\ell$ is the same
no matter what other jobs are in the system,
as we can simply ignore every moment in time
when neither $k$ nor $\ell$ is in service.
By considering the case where $k$ and $\ell$ are the only jobs in the system,
we find that the interference between them is
$\costP{\st{i_k}{x_k}, \st{i_\ell}{x_\ell}}$, as desired.

It remains to guess $\costP{\st{i_k}{x_k}, \st{i_\ell}{x_\ell}}$.
With only two jobs in the system,
queueing time is just the completion time of whichever job finishes first.
Under the policy of serving $k$ first no matter what,
the expected queueing time would be $\E{\prevailing{i_k}{x_k}}$
by~\fref{eq:cost_a_prevailing},
and similarly for serving~$\ell$ first.
It turns out that the optimal policy, improving on both of these options,
has expected queueing time
\begin{align*}
  \costP{\st{i_k}{x_k}, \st{i_\ell}{x_\ell}}
  &= \E{\min\{\prevailing{i_k}{x_k}, \prevailing{i_\ell}{x_\ell}\}} \\
  &= \integral{0}{\infty}
      \P{\prevailing{i_k}{x_k} > r
        \text{ and } \prevailing{i_\ell}{x_\ell} > r}
     \,dr \\
  &= \integral{0}{\infty}
      \dd{r}\costA{r}{\st{i_k}{x_k}}\nudgel
      \dd{r}\costA{r}{\st{i_\ell}{x_\ell}}
     \,dr.
\end{align*}
It suffices to verify our guess with the optimality equation \fref{eq:opt_p}.

\begin{proposition}
  \label{prop:cost_p}
  The cost function of BPOP~P is
  \begin{equation*}
    \costP{\sts}
    = \sum_{\mathclap{k < \ell}}\mkern5mu
        \integral{0}{\infty}
          \dd{r}\costA{r}{\st{i_k}{x_k}}\nudgel
          \dd{r}\costA{r}{\st{i_\ell}{x_\ell}}
        \,dr.
  \end{equation*}
\end{proposition}

\begin{proof}
  We will show the desired formula satisfies
  the optimality equation \fref{eq:opt_p}.
  Let
  \begin{equation*}
    \iftoggle{widecol}{%
      \cost[k\ell]{\dots, \st{i_k}{x_k}, \dots, \st{i_\ell}{x_\ell}, \dots}}{%
      \cost[k\ell]{\sts}}
    = \integral{0}{\infty}
        \dd{r}\costA{r}{\st{i_k}{x_k}}\nudgel
        \dd{r}\costA{r}{\st{i_\ell}{x_\ell}}
      \,dr
  \end{equation*}
  be the ``interference'' between $k$ and~$\ell$.
  It suffices to show it satisfies $\serve{k}\cost[k\ell]{\sts} \geq -1$,
  with equality for some~$k$ and all $\ell \neq k$.
  This suffices because $\serve{k}\cost[\ell m]{\sts} = 0$
  for all $\ell, m \neq k$, so
  \begin{equation*}
    \serve{k}\costP{\sts}
    = \sum_{\mathclap{l \neq k}}\serve{k}\cost[k\ell]{\sts}
    \geq -(n - 1),
  \end{equation*}
  achieving equality for some~$k$
  as required by the optimality equation \fref{eq:opt_p}.
  Using \fref{def:prevailing} and integration by parts, we compute
  \begin{align}
    \label{eq:integration_by_parts}
    \cost[k\ell]{\sts}
    \iftoggle{widecol}{%
      &= \integral{0}{\infty}
          \dd{r}\costA{r}{\st{i_k}{x_k}}\nudgel
          \dd{r}\costA{r}{\st{i_\ell}{x_\ell}}
        \,dr \\}{}
    &= \integral{0}{\infty}
        \P{\prevailing{i_\ell}{x_\ell} > r}
        \dd{r}\costA{r}{\st{i_k}{x_k}}
      \,dr \\
    &= \E{\costA{\prevailing{i_\ell}{x_\ell}}{\st{i_k}{x_k}}}.
  \end{align}
  We know the cost function of BPOP~A must satisfy \fref{eq:opt_a}, so,
  slightly abusing the $\serve{k}$ notation
  by applying it to $\costA{r}{\st{i_k}{x_k}}$,
  \begin{equation*}
    \serve{k}\cost[k\ell]{\sts}
    = \E{\serve{k}\costA{\prevailing{i_\ell}{x_\ell}}{\st{i_k}{x_k}}}
    \geq -1.
  \end{equation*}
  The job~$k$ for which equality holds is that of minimal SJP index.
  If $\fair{\st{i_k}{x_k}} \leq \fair{\st{i_\ell}{k_\ell}}$,
  then $\P{\fair{\st{i_k}{x_k}} \leq \prevailing{i_\ell}{x_\ell}} = 1$
  by~\fref{eq:prevailing_geq_fair}.
  This and \fref{def:fair} imply that with probability~$1$,
  serving~$k$ is the optimal action
  in BPOP~A with job~$k$
  and random expected bypass time $\prevailing{i_\ell}{x_\ell}$,
  so $\serve{k}\costA{\prevailing{i_\ell}{x_\ell}}{\st{i_k}{x_k}} = -1$,
  as desired.
\end{proof}

\subsection{Guessing the Cost Function: Arrivals}
\label{sub:guess}

We have successfully guessed the cost function for BPOP~P.
Our next task is to adapt our guess to BPOP~Q, which incorporates arrivals.
Similarly to how understanding BPOP~A helped with BPOP~P,
we first take some time to understand BPOPs B and~C.

We begin by relating BPOP~B, which has arrivals at rate~$\lambda$,
to BPOP~A, which has no arrivals.
In \fref{lem:cost_a}, we showed that the cost of BPOP~A
decomposes into a sum of independent single-job terms.
This decomposition works because
we can minimize the expected busy period length
by \emph{individually optimizing} the bypass policy for each job.
Our guess for BPOP~B is that, even with arrivals,
we can continue to optimize each job individually
when it comes to bypass decisions.
This makes the cost of BPOP~B the length of a busy period
whose load is given by the following definition.

\begin{definition}
  \label{def:load}
  The \emph{bypass-adjusted remaining size}
  of a job in state $\st{i}{x}$ with expected bypass time~$r$ is
  $\costA{r}{\st{i}{x}}$,
  and the \emph{bypass-adjusted load} is
  \begin{equation*}
    \load{r} = \lambda\costA{r}{\st{a}{0}}.
  \end{equation*}
  We recover the ordinary load as
  $\load{} = \lim_{r \to \infty}\load{r} = \lambda\E{\totalsize{\st{a}{0}}}$.
\end{definition}

BPOP~C is similar to BPOP~B, except arrivals stop while bypassing a job.
To guess the cost function of BPOP~C,
we consider that a busy period started by a bypass in BPOP~B takes time
\begin{equation*}
  \frozen{r} = \frac{r}{1 - \load{r}}.
\end{equation*}

\begin{lemma}
  \label{lem:cost_bc}
  The cost functions of BPOPs~B and~C satisfy
  \begin{equation*}
    \costB{r}{\sts}
    = \costC{\frozen{r}}{\sts}
    = \frac{\costA{r}{\sts}}{1 - \load{r}}.
  \end{equation*}
\end{lemma}

\proofin{app:proofs}

We are finally ready to address BPOP~Q.
Taking inspiration from our solution to BPOP~P in \fref{prop:cost_p},
we might hope that $\costQ{\sts}$ is a sum of ``interferences''
$\cost[k\ell]{\sts}$,
though the definition of interference should certainly change.
But we immediately see that it is not so simple,
because $\costQ{\st{a}{0}}$ is clearly nonzero
even though there is only one job in the system.
We instead guess that $\costQ{}$ has the form
\begin{equation*}
  \costQ{\sts}
  = \sum_k \costQ{\st{i_k}{x_k}}
    + \sum_{\mathclap{k < \ell}} \cost[k\ell]{\sts},
\end{equation*}
where, as in $\costP{}$, the interference $\cost[k\ell]{\sts}$
depends only on jobs $k$ and~$\ell$.
While the interference in BPOP~P is in terms of BPOP~A,
we guess that the interference in BPOP~Q is in terms of BPOP~C,
which has arrivals during service.

\begin{proposition}
  \label{prop:cost_q}
  The cost function of BPOP~Q satisfies
  \begin{align}
    \label{eq:cost_q_one}
    \iftoggle{widecol}{}{&}
    \costQ{\st{i}{x}}
    \iftoggle{widecol}{&}{}
    = \load{}\costQ{\st{a}{0}}
    \iftoggle{widecol}{\\ &\quad}{\\ &\qqquad}
      + \integral{x}{\infty}
        \integral{0}{\infty}
          \lambda \frac{\tail[i]{y}}{\tail[i]{x}}\nudgec
          \dd{s}\costC{s}{\st{i}{y}}\nudgel
          \dd{s}\costC{s}{\st{a}{0}}
        \,ds
      \,dy
    \iftoggle{widecol}{\\ &\quad}{\\ &\qqquad}
      + \sum_{\mathclap{j \neq i}}\mkern3mu
        \integral{0}{\infty}
          \integral{0}{\infty}
            \lambda q_{ij} \tail[j]{y}\nudgel
            \dd{s}\costC{s}{\st{j}{y}}\nudgel
            \dd{s}\costC{s}{\st{a}{0}}
          \,ds
        \,dy,
  \end{align}
  where $q_{ij}\esub$ is the probability that
  a job starting in stage~$i$ at some point reaches stage~$j$,
  and
  \iftoggle{widecol}{\begin{equation}}{\begin{align}}
    \label{eq:cost_q_many}
    \costQ{\sts}
    \iftoggle{widecol}{}{&}
    = \sum_k\costQ{\st{i_k}{x_k}}
    \iftoggle{widecol}{}{\\ &\quad}
      + \sum_{\mathclap{k < \ell}}\mkern5mu
        \integral{0}{\infty}
          \dd{s}\costC{s}{\st{i_k}{x_k}}\nudgel
          \dd{s}\costC{s}{\st{i_\ell}{x_\ell}}
        \,ds.
  \iftoggle{widecol}{\end{equation}}{\end{align}}
\end{proposition}

To prove \fref{prop:cost_q},
we need some auxiliary definitions and lemmas.
Let
\begin{equation*}
  \iftoggle{widecol}{%
    \cost[k\ell]{\dots, \st{i_k}{x_k}, \dots, \st{i_\ell}{x_\ell}, \dots}}{%
    \cost[k\ell]{\sts}}
  = \integral{0}{\infty}
      \dd{s}\costC{s}{\st{i_k}{x_k}}\nudgel
      \dd{s}\costC{s}{\st{i_\ell}{x_\ell}}
    \,ds
\end{equation*}
be the interference between jobs $k$ and~$\ell$.
We can rewrite \fref{eq:cost_q_one} and \fref{eq:cost_q_many}
in terms of interferences.
As in the proof of \fref{prop:cost_p},
the key to proving \fref{prop:cost_q}
is to write the interference between jobs $k$ and~$\ell$
as a minimum of random variables.
Just as interference in BPOP~P uses a random variable defined using BPOP~A,
interference in BPOP~Q uses a random variable defined using BPOP~C.

\begin{definition}
  The \emph{prevailing index with arrivals} of a job in state $\st{i}{x}$,
  denoted~$\freezing{i}{x}$
  is a random variable with tail function
  \begin{equation}
    \label{eq:freezing}
    \P{\freezing{i}{x} > s} = \dd{s}\costC{s}{\st{i}{x}}.
  \end{equation}
  Like the prevailing index without arrivals of \fref{def:prevailing},
  the variant with arrivals
  is a computational tool used for its convenient distribution,
  and we assume pairwise independence.
\end{definition}

\begin{lemma}
  \label{lem:freezing_defined}
  For any state $\st{i}{x}$,
  the prevailing index with arrivals $\freezing{i}{x}\esub$ is a well defined.
  That is, $\dd{s}\costC{s}{\st{i}{x}}$ is a valid tail function which is
  decreasing in~$s$ and bounded between $0$ and~$1$.
\end{lemma}

\begin{lemma}
  \label{lem:freezing_fair}
  For any state $\st{i}{x}$,
  \begin{equation*}
    \P{\freezing{i}{x} \geq \frozen{\fair{\st{i}{x}}}} = 1.
  \end{equation*}
\end{lemma}

\proofin[Proofs of Lemmas~\ref{lem:freezing_defined}
  and~\ref{lem:freezing_fair}]{app:proofs}

We are ready for the main technical lemma behind \fref{prop:cost_q}.
Its proof is similar to the last part of the proof of \fref{prop:cost_p}.

\begin{lemma}
  \label{lem:opt_interference}
  For any jobs $k$ and~$\ell$ with
  $\fair{\st{i_k}{x_k}} \leq \fair{\st{i_\ell}{x_\ell}}\esub$,
  \begin{equation*}
    (\serve{\ell} + \arrive)\cost[k\ell]{\sts}
    \geq (\serve{k} + \arrive)\cost[k\ell]{\sts}
    = -1.
  \end{equation*}
\end{lemma}

\proofin{app:proofs}

Our final lemma is a general claim about the serve-$k$ operator.

\begin{lemma}
  \label{lem:serve_one}
  For any function $\alpha$ assigning
  a cost rate $\alpha(\st{i}{x})$ to each state~$\st{i}{x}$,
  the equation
  $\serve{1}\cost{\st{i}{x}} = -\alpha(\st{i}{x})$
  is satisfied by
  \begin{equation}
    \label{eq:cost_one}
    \cost{\st{i}{x}}
    = \integral{x}{\infty}
        \frac{\tail[i]{y}}{\tail[i]{x}}
        \alpha(\st{i}{y})
      \,dy
      + \sum_{\mathclap{j \neq i}}\mkern3mu
        \integral{0}{\infty}
          q_{ij} \tail[j]{y}
          \alpha(\st{j}{y})
        \,dy,
  \end{equation}
  where $q_{ij}\esub$ is the probability that
  a job starting in stage~$i$ at some point reaches stage~$j$.
\end{lemma}

\proofin{app:proofs}

A notable special case of \fref{lem:serve_one}
is that of constant functions $\alpha(\st{i}{x}) = \alpha$,
in which case $\serve{1}\cost{\st{i}{x}} = -\alpha$ is satisfied by
\begin{equation*}
  \cost{\st{i}{x}} = \alpha\E{\totalsize{\st{i}{x}}},
\end{equation*}
because \fref{eq:cost_one} with constant~$\alpha$
is $\alpha$ times the sum of
the expected remaining service times in each state.

\begin{proof}[\proofof{prop:cost_q}]
  First, note that \fref{eq:cost_q_one} and \fref{eq:cost_q_many}
  completely determine $\costQ{}$.
  Specifically, \fref{eq:cost_q_one} implies
  \begin{equation}
    \label{eq:cost_q_init}
    \costQ{\st{a}{0}}
    = \frac{1}{1 - \load{}}
      \sum_{\mathclap{i \in \stages}}\mkern3mu
        \integral{0}{\infty}
          \lambda q_{ai} \tail[i]{x}
          \cost[1 \mkern0.5mu 2]{\st{i}{x}, \st{a}{0}}
        \,dx,
  \end{equation}
  This determines $\costQ{\st{i}{x}}$ by~\fref{eq:cost_q_one},
  which in turn determines $\costQ{\sts}$ by~\fref{eq:cost_q_many}.
  Therefore, it suffices to show that \fref{eq:opt_q} is satisfied
  assuming \fref{eq:cost_q_one} and \fref{eq:cost_q_many} hold.
  By~\fref{eq:cost_q_many},
  \begin{align*}
    \arrive\costQ{\dots}
    &= \lambda(\costQ{\dots, \st{a}{0}} - \costQ{\dots}) \\
    &= \lambda\costQ{\st{a}{0}}
      + \sum_{\mathclap{k \leq n}} \lambda\cost[k(n+1)]{\dots, \st{a}{0}}.
  \end{align*}
  This means that \fref{eq:opt_q} in the case of $\sts = (\st{i}{x})$
  becomes
  \begin{align*}
    \label{eq:opt_q_one}
    \serve{1}\costQ{\st{i}{x}}
    &= -\arrive\costQ{\st{i}{x}} \\
    &= -\lambda\costQ{\st{a}{0}}
      - \lambda\cost[1 \mkern0.5mu 2]{\st{i}{x}, \st{a}{0}},
  \end{align*}
  which holds by \fref{lem:serve_one} applied to \fref{eq:cost_q_one}.
  It remains only to prove \fref{eq:opt_q} in the case of multiple jobs.
  Once more abusing the $\serve{k}$ notation, we compute
  \iftoggle{widecol}{\begin{equation}}{\begin{align}}
    \label{eq:serve_q_many}
    (\serve{k} + \arrive)\costQ{\sts}
    = (\serve{k} + \arrive)\costQ{\st{i_k}{x_k}}
      + \sum_{\mathclap{l \neq k}}(\serve{k} + \arrive)\cost[k\ell]{\sts}.
  \iftoggle{widecol}{\end{equation}}{\end{align}}
  By \fref{eq:opt_q_one},
  we know $(\serve{k} + \arrive)\costQ{\st{i_k}{x_k}} = 0$,
  and by \fref{lem:opt_interference},
  we know $(\serve{k} + \arrive)\cost[k\ell]{\sts} \geq -1$.
  Combined with~\fref{eq:serve_q_many}, this yields
  \begin{equation*}
    (\serve{k} + \arrive)\costQ{\sts} \geq -(n - 1).
  \end{equation*}
  To prove \fref{eq:opt_q},
  it remains only to show that equality holds above for some job~$k$.
  By \fref{lem:opt_interference},
  this is the case when $k$ is the job of minimal SJP index,
  as then $(\serve{k} + \arrive)\cost[k\ell]{\sts} = -1$ for all $\ell \neq k$.
\end{proof}

\begin{proof}[\proofof{thm:queueing_time}]
  The desired expression for $\E{\Tqueue}$
  is immediate from \fref{eq:renewal_reward} and~\fref{eq:cost_q_init}.
\end{proof}

\section{Practical Takeaways for Scheduling Multistage Jobs}
\label{sec:takeaways}

We have thus far shown how to compute the Gittins index of a multistage job
(\fref{sec:composition})
and analyzed the performance of the resulting Gittins policy
(\fref{sec:analysis}).
In this section we take a step back and consider
the broader implications of this work for system designers.

\subsection{Rough Guideline: Prioritize Variable Stages}
\label{sub:prioritize_variable}

Even for single-stage jobs,
the exact Gittins policy is impossible to pithily describe.
At best, one can summarize it with a rough guideline:
prioritize jobs that have a chance of being short.
The story is even more complicated for multistage jobs.
Fortunately, we can give another guideline for this case:
given two multistage jobs of similar size distributions,
\emph{prioritize the job which frontloads variability}.
We can use \fref{thm:composition} to rigorously prove
a specific illustrating case of this guideline.

\begin{proposition}
  \label{prop:stochastic_first}
  Let $\single{i}$ be a single-stage job with deterministic size
  $\size{i} = d$.
  Then for any job~$\J$,
  the Gittins policy prioritizes
  $\J \seq \single{i}$ over $\single{i} \seq \J$ because
  \begin{equation*}
    \fair{\J \seq \single{i}}
    = \fair{\J} + d
    \leq \fair{\single{i} \seq \J}.
  \end{equation*}
\end{proposition}

\proofin{app:proofs}

In fact, the result generalizes to the case where $\size{i}$ has the so-called
\emph{new better than used in expectation} (NBUE) property
\citep{m/g/1_gittins_aalto}.

Consider jobs $\jt{A}$, $\jt{B}$, and $\jt{C}$
from \fref{fig:multistage_difference}.
By \fref{prop:stochastic_first},
we should prioritize $\jt{B}$ first, followed by $\jt{C}$ and then~$\jt{A}$.
Intuitively, we prefer $\jt{B}$ because having the stochastic stage earlier
means that we learn information about the job's total size earlier.

\subsection{Impact of Exploiting Multistage Structure}
\label{sub:impact}

It is possible to schedule multistage jobs
by using a ``blind'' algorithm that ignores the multistage structure.
That is, a blind policy makes scheduling decisions
using only each job's \emph{total age},
the total amount of service the job has received so far.
We might wonder:
are blind policies ``good enough'' to make it
not worth worrying about multistage jobs?
We show below that the answer is at least sometimes no:
exploiting multistage structure can
\emph{significantly decrease mean response time}
compared to the best blind policy.

We consider three different scheduling policies in this section:
\begin{itemize}
\item
  \emph{First-come, first-served} (FCFS) ignores multistage structure:
  it serves jobs to completion in the order they arrive.
\item
  The \emph{blind Gittins policy} (BGP) ignores multistage structure:
  it computes each job's Gittins index based on only its total age
  and always serves the job of maximal Gittins index.
\item
  Our policy, the \emph{multistage Gittins policy} (MGP),
  exploits multistage structure:
  it computes each job's Gittins index based on
  its current stage and age within that stage
  and always serves the job of maximal Gittins index.
  This computation is enabled by \fref{alg:gittins}.
\end{itemize}
Figures~\ref{fig:q} and~\ref{fig:r},
explained below,
show two examples in which MGP is far superior to BGP
with respect to mean response time.

\begin{figure}
  \iftoggle{widecol}{\begin{minipage}[t]{0.48\linewidth}}{}
  \centering
  \includegraphics[width=\iftoggle{widecol}{}{0.8}\linewidth]{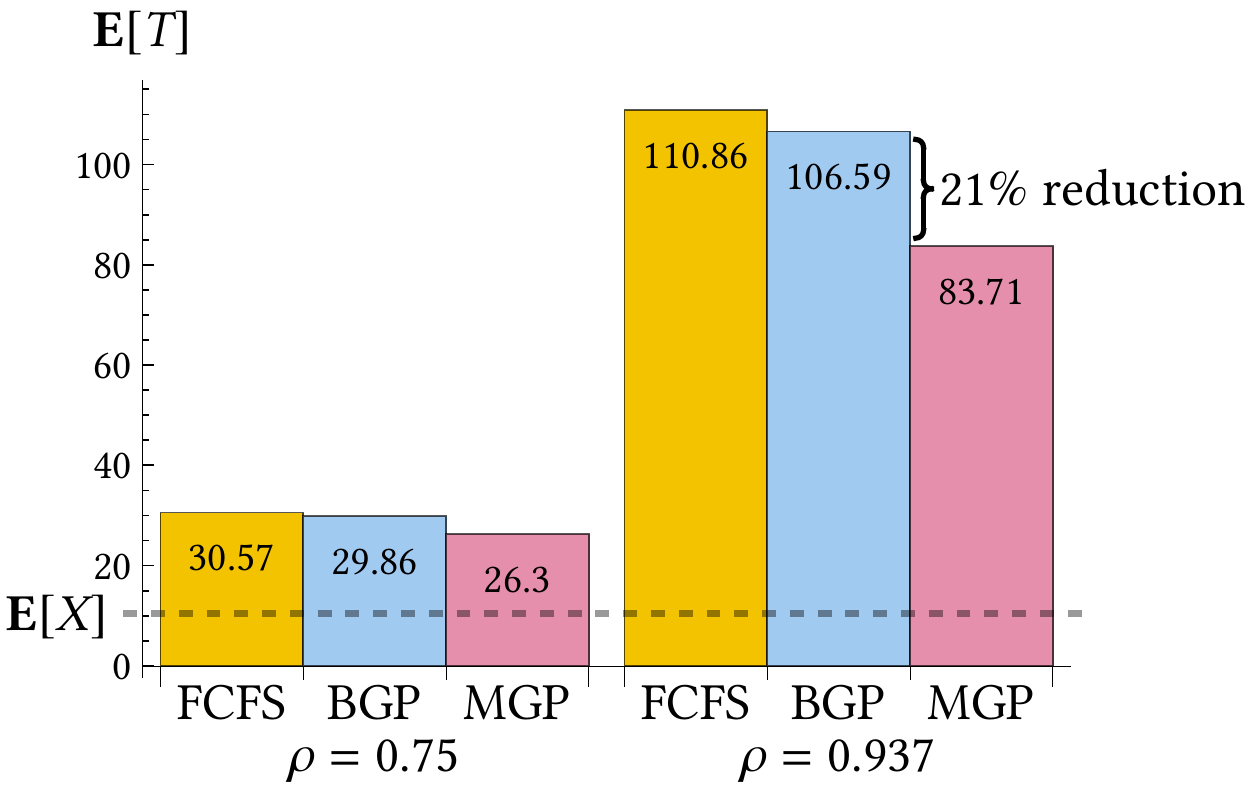}
  \squish

  \caption{Impact of Known Stage Order}
  \label{fig:q}
\iftoggle{widecol}{%
  \end{minipage}\hfill\begin{minipage}[t]{0.48\linewidth}}{%
  \end{figure}\begin{figure}}
  \centering
  \includegraphics[width=\iftoggle{widecol}{}{0.8}\linewidth]{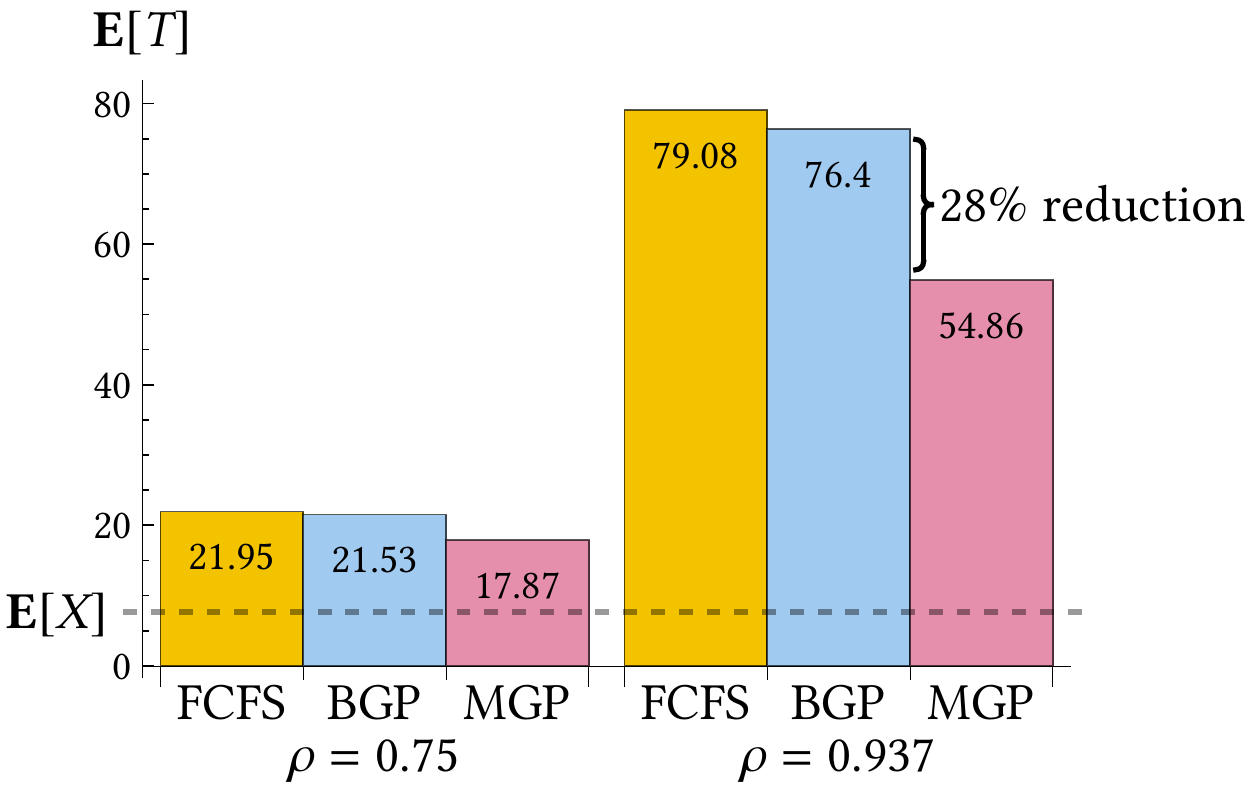}
  \squish

  \caption{Impact of Early Size Information}
  \label{fig:r}
  \iftoggle{widecol}{\end{minipage}}{}
\end{figure}

The first system demonstrates the benefit of
knowing the order of a job's stages.
We assume jobs are a mixture of types $\jt{A}$, $\jt{B}$, and $\jt{C}$
from \fref{fig:multistage_difference},
each with equal probability.
Each job has two stages of deterministic size $d = 2$
and one stage of stochastic size~$S$,
which is chosen uniformly from $\{1, 12\}$.
The different job types correspond to
different orderings of these three stages.
\begin{itemize}
\item
  BGP sees every job as having size distribution $S + 2d$.
\item
  MGP differentiates jobs based on stage order,
  which guides its scheduling decisions
  as shown in \fref{sub:prioritize_variable}.
\end{itemize}
\fref{fig:q} shows MGP reduces mean response time compared to BGP
by 12\% at moderate load and 21\% at high load.

The second system demonstrates the benefit of
learning information about a job's size early.
The jobs have type~$\jt{R}$ from \fref{fig:multistage_examples}:
they have a diagnosis stage of size $D = 1$
followed by either an easy repair stage of size $R_{\text{easy}} = 4$
or hard repair stage of size $R_{\text{hard}} = 12$.
Repair is easy with probability $p = 2/3$.
\begin{itemize}
\item
  BGP does not directly learn whether a job is easy or hard to repair,
  though it can infer it is hard once a job reaches age~$5$,
  at which point it gives the job lower priority
  than jobs which have not been run yet.
\item
  MGP learns whether a job has easy or hard repair
  immediately after diagnosis,
  allowing it to give jobs with hard repairs lower priority sooner.
\end{itemize}
\fref{fig:r} shows MGP reduces mean response time compared to BGP
by 17\% at moderate load and 28\% at high load.

\section{Conclusion}

In this paper, we introduce \emph{multistage jobs},
a tool for modeling scheduling problems in which
the scheduler has partial but incomplete information
about each job's remaining size.
The multistage job model is more general than
the standard M/G/1 model with unknown job sizes.
The optimal scheduling algorithm for the standard M/G/1,
namely the Gittins policy,
is complex but tractable.
Using a new formulation of the Gittins index
called \emph{single-job profit} (SJP),
we are able to reduce the task of optimally scheduling multistage jobs
to the task of solving SJP for the job's individual stages.
Finally, we leverage SJP to provide a \emph{closed-form analysis}
of the Gittins policy for multistage jobs,
which we use to demonstrate the importance of exploiting multistage structure.

\bibliographystyle{ACM-Reference-Format}
\bibliography{refs}

\appendix

\section{Formal Definitions of Job Type Operations}
\label{app:definitions}

We write $\sqcup$ for disjoint union
and let $\stages[\J]' = \stages[\J] \setminus \{z\}$.

\begin{definition}
  The \emph{state-conditioned job type} for job type~$\J$ and state~$\st{i}{x}$,
  denoted $\st[\J]{i}{x}$,
  is the job type obtained by starting a type~$\J$ job in state $\st{i}{x}$.
  Formally, letting $i'$ be a new stage label, we define
  \begin{align*}
    \stages[{\st[\J]{i}{x}}] &= (\stages[\J] \setminus \{i\}) \sqcup \{i'\} \\
    \init[{\st[\J]{i}{x}}] &= i' \\
    \fin[{\st[\J]{i}{x}}] &= \fin[\J] \\
    \trans[{\st[\J]{i}{x}}]{j}{k} &=
    \begin{dcases}
      \trans[\J]{i}{k} & \text{if } j = i' \\
      \trans[\J]{j}{k} & \text{otherwise} \\
    \end{dcases} \\
    \tail[i']{t} &= \frac{\tail[i]{x + t}}{\tail[i]{x}}.
  \end{align*}
\end{definition}

\begin{definition}
  The \emph{sequential composition} of job types $\J$ and~$\K$,
  denoted $\J \circ \K$,
  is the job type obtained by ``stitching together'' $\J$ and $\K$.
  Formally, we define
  \begin{align*}
    \stages[\J \seq \K] &= \stages[\J]' \sqcup \stages[\K] \\
    \init[\J \seq \K] &= \init[\J] \\
    \fin[\J \seq \K] &= \fin[\K] \\
    \trans[\J \seq \K]{i}{j} &=
    \begin{dcases}
      \trans[\J]{i}{j} & \text{if } i, j \in \stages[\J]' \\
      \trans[\J]{i}{\fin[\J]}
      & \text{if } i \in \stages[\J]' \text{ and } j = \init[\K] \\
      \trans[\K]{i}{j}
      & \text{if } i, j \in \stages[\K] \\
      0 & \text{otherwise.}
    \end{dcases}
  \end{align*}
\end{definition}

\begin{definition}
  The \emph{single-stage job type} with stage~$i$, denoted~$\single{i}$,
  is the job type with just one stage~$i$ of size~$\size{i}$.
  Formally, we define
  $\stages[\single{i}] = \{\init[\single{i}]\} = \{i\}$
  and $\trans[\single{i}]{i}{\fin} = 1$.
\end{definition}

\begin{definition}
  Let $\mc{L}$ be a finite set.
  The \emph{mixture composition} of job types $\{\J_\ell\}_{\ell \in \mc{L}}$
  with probabilities $\{q_\ell\}_{\ell \in \mc{L}}$
  is the job type representing a job with a randomly chosen type,
  choosing $\J_\ell$ with probability~$q_\ell$.
  We denote the mixture composition by
  $\{\J_\ell \text{ w.p. } q_\ell\}_{\ell \in \mc{L}}$,
  abbreviating to $\J_{\mc{L}}$
  when the probabilities are unambiguous.
  Formally, letting $\init'$ and $\fin$ be new stage labels, we define
  \begin{align*}
    \stages[\J_{\mc{L}}]
    &= \{\init', \fin'\}
      \sqcup \bigsqcup_{\mathclap{\ell \in \mc{L}}} \stages[\J_\ell]' \\
    \init[\J_{\mc{L}}] &= \init' \\
    \fin[\J_{\mc{L}}] &= \fin' \\
    \trans[\J_{\mc{L}}]{i}{j} &=
    \begin{dcases}
      q_\ell
      & \text{if } i = \init' \text{ and } j = \init[\J_\ell] \\
      \trans[\J_\ell]{i}{\fin[\J_\ell]}
      & \text{if } i \in \stages[\J_\ell]'
        \text{ for some } \ell \in \mc{L}
        \text{ and } j = \fin' \\
      \trans[\J_\ell]{i}{j}
      & \text{if } i, j \in \stages[\J_\ell]'
        \text{ for some } \ell \in \mc{L} \\
      0 & \text{otherwise}
    \end{dcases} \\
    \size{\init'} &= 0.
  \end{align*}
\end{definition}

\section{Algorithm Details for Discrete Stage Size Distributions}
\label{app:discrete}

To finish the algorithm given in \fref{sub:discrete},
it remains to
\begin{itemize}
\item
  compute $\profit{\single{i}}{}$ in $O(|X_i|)$ time and
\item
  compose
  $\profit{\K}{} \circ \profit{\single{i}}{} \circ \profit{\L}{}$
  in $O(|\K| + |X_i| + |\L|)$ time.
\end{itemize}
Both of these are made possible by representing the SJP functions
by linked lists.
This is possible because a job with finite state space
has finitely many stopping policies~$\pi$,
so by \fref{def:sjp},
the SJP function of a job with finite state space
is the maximum of finitely many linear functions.
Specifically, given
\begin{align*}
  0 = s_0 \leq s_1 &\leq \dots \leq s_n \\
  0 = \sigma_0 \leq \sigma_1 &\leq \dots \leq \sigma_n = 1,
\end{align*}
the list $((s_1, \sigma_1), \dots, (s_n, \sigma_n))$ represents
the following function,
which has $V(0) = 0$ and slope $\sigma_k$ over the interval $(s_k, s_{k + 1})$:
\begin{equation*}
  V(r) =
  \begin{dcases}
    0 & r < s_1 \\
    \dots \\
    \sigma_k(r - s_k)
    + \sum_{\ell = 1}^{\mathclap{k - 1}} \sigma_\ell(s_{\ell + 1} - s_\ell)
    & s_k \leq r < s_{k + 1} \\
    \dots \\
    r - s_n
    + \sum_{\ell = 1}^{\mathclap{n - 1}} \sigma_\ell(s_{\ell + 1} - s_\ell)
    & r > s_n.
  \end{dcases}
\end{equation*}

With this representation,
we can compute the composition of two functions
by traversing the two lists.
Each node of one of the input lists transforms to become
a node of the output list,
with the ordering given by interleaving the input lists.
Both the transformation and interleaving can be determined
in a manner similar to the ``merge'' step of merge sort:
we examine the first items of both lists,
choose one of them to transform,
add the transformed item to the output list,
remove the original item from its input list,
and repeat.
This interleaving takes linear time, as desired.

It remains only to compute a single-stage profit function in linear time.
Consider the single-stage job $\single{i}$ with size distribution
\begin{equation*}
  \size{i} =
  \begin{dcases}
    x_1 & \text{with probability } p_1 \\
    \dots \\
    x_n & \text{with probability } p_n.
  \end{dcases}
\end{equation*}
It is convenient to set $x_0 = 0$ and $p_0 = 0$,
and we assume for ease of exposition that $x_1 > 0$.
We begin by precomputing
\begin{align*}
  P_k
  &= \P{\size{i} > x_k}
  = 1 - \sum_{\ell = 0}^{\mathclap{k - 1}} p_\ell \\
  E_k
  &= \E{\min\{\size{i}, x_k\}}
  = \sum_{\ell = 0}^{\mathclap{k - 1}} P_\ell(x_{\ell + 1} - x_\ell)
\end{align*}
for $0 \leq k \leq n$, which takes $O(n)$ time.
To clarify, $P_0 = 0$ and $E_0 = 0$.
For brevity,
we write $\fair{k} = \fair{\st[\single{i}]{i}{x_\ell}}$.
In this finite-support case,
\fref{def:gittins} and \fref{prop:gittins_fair} imply
\begin{equation*}
  \fair{k} = \min_{\mathclap{m > k}} \frac{E_m - E_k}{P_m - P_k}.
\end{equation*}
We denote the minimizing~$m$ by~$m(k)$.
For any $k < \ell$,
it is straightforward to show that the following conditions are equivalent:
\begin{itemize}
\item
  $\fair{k} \leq \fair{\ell}$,
\item
  $m(k) \leq \ell$, and
\item
  $(E_\ell - E_k)/(P_\ell - P_k) \leq \fair{\ell}$.
\end{itemize}
The third inequality means that if we know $\fair{\ell}$,
we can obtain the results of the first two comparisons in $O(1)$ time.

Imagine that $1, \dots, n$ are people standing in a line looking to the right,
and suppose that person~$k$ has height~$\fair{k}$.
We say that person~$k$ \emph{sees} person $\ell > k$
if $\fair{m} \leq \fair{k} \leq \fair{\ell}$
for all $k < m < \ell$.
That is, $k$ sees $\ell$ if $\ell$ is taller than~$k$
and there is nobody even taller in the way.

Starting with the base case $\fair{n - 1} = x_n - x_{n - 1}$,
we now compute $\fair{k}$ for $k = n - 1, \dots, 0$.
We maintain a stack such that after we have computed $\fair{k + 1}$,
the stack contains each person~$\ell$ that $k + 1$ can see
in ascending order.
To compute $\fair{k}$,
we pop people off the stack
until $k$ can see the first person~$\ell$ on the stack,
at which point we push $k$ onto the stack.
We have defined seeing such that $m(k) = \ell$,
from which we obtain $\fair{k}$.
It is simple to see that we maintain the invariant.
Each person can only be pushed onto and popped off of the stack once,
so this entire process takes $O(n)$ time.
Given $\fair{0}, \dots, \fair{n - 1}$,
we can obtain $\profit{0}$ in $O(n)$ time
using the approach described at the end of \fref{sub:single-stage}.

\section{Deferred Proofs}
\label{app:proofs}

\begin{proof}[\proofof{lem:cost_bc}]
  Taking as given that
  \begin{equation*}
    \costB{r}{\sts}
    = \costC{\frozen{r}}{\sts}
    = \frac{\costA{r}{\sts}}{1 - \load{r}},
  \end{equation*}
  we show
  \begin{itemize}
  \item
    $\costB{r}{}$ satisfies the optimality equation \fref{eq:opt_b} and
  \item
    $\costC{\frozen{r}}{}$ satisfies
    the optimality equation \fref{eq:opt_c},
    using $\frozen{r}$ as the expected bypass time.
  \end{itemize}

  We start by verifying \fref{eq:opt_b}.
  Using \fref{lem:cost_a}, we compute
  \begin{align}
    \label{eq:arrive_b}
    \arrive\costB{r}{\dots}
    &= \frac{\arrive\costA{r}{\dots}}{1 - \load{r}} \\
    &= \frac{%
        \lambda(\costA{r}{\dots, \st{a}{0}} - \costA{r}{\dots})}{%
        1 - \load{r}} \\
    &= \frac{\load{r}}{1 - \load{r}}.
  \end{align}
  Because $\costA{r}{}$ satisfies \fref{eq:opt_a},
  by pulling out the common $\arrive\costB{r}{\sts}$ terms below,
  we obtain
  \begin{align*}
    \iftoggle{widecol}{}{&}
    \min_k \min\biggl\{
    \begin{matrix}
      (\serve{k} + \arrive)\costB{r}{\sts} \\
      (\bypass{r}{k} + \arrive)\costB{r}{\sts}
    \end{matrix}
    \biggr\}
    \iftoggle{widecol}{&}{\\ &\qquad}
    = \min_k \min\biggl\{
    \begin{matrix}
      \serve{k}\costB{r}{\sts} \\
      \bypass{r}{k}\costB{r}{\sts}
    \end{matrix}
    \biggr\} + \frac{\load{r}}{1 - \load{r}}
    \iftoggle{widecol}{\\ &}{\\ &\qquad}
    = \frac{1}{1 - \load{r}} \min_k \min\biggl\{
    \begin{matrix}
      \serve{k}\costA{r}{\sts} \\
      \bypass{r}{k}\costA{r}{\sts}
    \end{matrix}
    \biggr\} + \frac{\load{r}}{1 - \load{r}}
    \iftoggle{widecol}{\\ &}{\\ &\qquad}
    = -1,
  \end{align*}
  which, as desired, is \fref{eq:opt_b}.

  We now turn to verifying \fref{eq:opt_c}.
  Because $\costB{r}{\sts}$ satisfies \fref{eq:opt_b},
  we know $\costC{\frozen{r}}{\sts}$ satisfies the ``top branch''
  of \fref{eq:opt_c}.
  It thus suffices to show the ``bottom branch''
  \begin{equation*}
    \bypass{\frozen{r}}{k}\costB{\frozen{r}}{\sts} \geq -1,
  \end{equation*}
  with equality if $(\bypass{r}{k} + \arrive)\costB{\frozen{r}}{\sts} = -1$.
  Recalling~\fref{eq:arrive_b}, we compute
  \begin{align*}
    \iftoggle{widecol}{}{&}
    \bypass{\frozen{r}}{k}\costB{r}{\sts}
    \iftoggle{widecol}{&}{\\ &\qquad}
    = \bypass{}{k}\biggl[\frac{r}{1 - \load{r}}\biggr]
      \costB{r}{\sts}
    \iftoggle{widecol}{\\ &}{\\ &\qquad}
    = \frac{1 - \load{r}}{r}
      (\costB{r}{\dots, \st{\fin}{0}, \dots}
       - \costB{r}{\dots, \st{i_k}{x_k}, \dots})
    \iftoggle{widecol}{\\ &}{\\ &\qquad}
    = (1 - \load{r})\bypass{r}{k}\costB{r}{\sts}
    \iftoggle{widecol}{\\ &}{\\ &\qquad}
    = (1 - \load{r})(\bypass{r}{k} + \arrive)\costB{r}{\sts} - \load{r}.
    \iftoggle{widecol}{\\ &}{\\ &\qquad}
    \geq -1,
  \end{align*}
  with equality if $(\bypass{r}{k} + \arrive)\costB{r}{\sts} = -1$, as desired.
\end{proof}

\begin{proof}[\proofof{lem:freezing_defined}]
  We follow much the same argument as for \fref{lem:sjp_properties}.
  It is clear that $\costC{0}{\st{i}{x}} = 0$,
  so it suffices to show that $\costC{s}{\st{i}{x}}$ is
  \begin{enumerate}[(i)]
  \item
    \label{item:cost_c_concave_nonincreasing}
    concave and nondecreasing in~$s$,
    which implies decreasing derivative bounded below by~$0$;
    and
  \item
    \label{item:cost_c_bounded_above}
    bounded above by $s$,
    which implies derivative bounded above by~$1$.
  \end{enumerate}
  Given a fixed policy for BPOP~C, the expected cost is
  \begin{equation*}
    s\E{\text{number of bypassed jobs}} + \E{\text{time spent on all jobs}},
  \end{equation*}
  which is concave and nonincreasing in~$s$.
  The cost $\costC{s}{\st{i}{x}}$ is the infimum of all such functions,
  implying claim~\fref{item:cost_c_concave_nonincreasing}.
  A possible policy is to bypass the job immediately,
  which has expected cost~$s$,
  implying claim~\fref{item:cost_c_bounded_above}.
\end{proof}

\begin{proof}[\proofof{lem:freezing_fair}]
  By \fref{lem:cost_bc}, if $r \leq \fair{\st{i}{x}}$,
  \begin{align*}
    \costC{\frozen{r}}{\st{i}{x}}
    &= \frac{\costA{r}{\st{i}{x}}}{1 - \load{r}} \\
    &= \frac{r}{1 - \load{r}} \\
    &= \frozen{r}.
  \end{align*}
  This means that for all $s < \frozen{\fair{\st{i}{x}}}$,
  \begin{equation*}
    \P{\freezing{i}{x} > s}
    = \dd{s} \costC{s}{\st{i}{x}}
    = 1,
  \end{equation*}
  so $\P{\freezing{i}{x} \geq \frozen{\fair{\st{i}{x}}}} = 1$.
\end{proof}

\begin{proof}[\proofof{lem:opt_interference}]
  By the same reasoning as~\fref{eq:integration_by_parts},
  \begin{equation*}
    \cost[k\ell]{\sts} = \E{\costC{\freezing{i_\ell}{x_\ell}}{\st{i_k}{x_k}}}.
  \end{equation*}
  Because $\costC{s}{\sts}$ satisfies \fref{eq:opt_c},
  \begin{equation}
    \label{eq:serve_k_interference}
    (\serve{k} + \arrive)\cost[k\ell]{\sts}
    = \E{(\serve{k} + \arrive)\costC{\freezing{i_\ell}{x_\ell}}{\st{i_k}{x_k}}}
    \geq -1,
  \end{equation}
  where we continue to slightly abuse the $\serve{k}$ notation.
  The same argument holds with $k$ and $\ell$ swapped.
  It remains only to show that \fref{eq:serve_k_interference}
  is actually an equality,
  which requires the
  $\fair{\st{i_k}{x_k}} \leq \fair{\st{i_\ell}{x_\ell}}$ hypothesis.
  By \fref{lem:freezing_fair}, we have
  $\P{\frozen{\fair{\st{i_k}{x_k}}} \leq \freezing{i_\ell}{x_\ell}} = 1$.
  Thus, to prove \fref{eq:serve_k_interference} is an equality,
  it suffices to show that for any $r \geq \fair{\st{i_k}{x_k}}$,
  \begin{equation*}
    (\serve{k} + \arrive)\costC{\frozen{r}}{\st{i_k}{x_k}} = -1,
  \end{equation*}
  because $\frozen{r}$ for $r \geq \fair{\st{i_k}{x_k}}$
  covers all possible values of $\freezing{i_\ell}{x_\ell}$.
  Using \fref{eq:arrive_b} and \fref{lem:cost_bc},
  we compute
  \begin{align*}
    \iftoggle{widecol}{}{&}
    (\serve{k} + \arrive)\costC{\frozen{r}}{\st{i_k}{x_k}}
    \iftoggle{widecol}{&}{\\ &\qquad}
    = (\serve{k} + \arrive)\costB{r}{\st{i_k}{x_k}}
    \iftoggle{widecol}{\\ &}{\\ &\qquad}
    = \serve{k}\costB{r}{\st{i_k}{x_k}} + \frac{\load{r}}{1 - \load{r}}
    \iftoggle{widecol}{\\ &}{\\ &\qquad}
    = \frac{1}{1 - \load{r}}\serve{k}\costA{r}{\st{i_k}{x_k}}
      + \frac{\load{r}}{1 - \load{r}}.
  \end{align*}
  Because $r \geq \fair{\st{i_k}{x_k}}$,
  it is optimal to serve job~$k$ in BPOP~A with expected bypass time~$r$.
  This means $\serve{k}\costA{r}{\st{i_k}{x_k}} = -1$,
  which completes the computation as desired.
\end{proof}

\begin{proof}[\proofof{lem:serve_one}]
  For all stages $j \neq i$,
  the probability $q_{ij}$ satisfies
  $q_{ij} = \sum_{k \in \stages}\trans{i}{k}q_{kj}$,
  which means
  \begin{equation}
    \label{eq:cost_stage_diff}
    \sum_{j \neq i}\trans{i}{j}\cost{\st{j}{0}} - \cost{\st{i}{x}}
    = -\integral{x}{\infty}
        \frac{\tail[i]{y}}{\tail[i]{x}}
        \alpha(\st{i}{y})
      \,dy.
  \end{equation}
  The second term of $\cost{\st{i}{x}}$ is constant in~$x$,
  so using the Leibniz integral rule, we compute
  \begin{align*}
    \iftoggle{widecol}{}{&}
    \dd{x}\cost{\st{i}{x}}
    \iftoggle{widecol}{&}{\\ &\qquad}
    = \dd{x} \integral{x}{\infty}
        \frac{\tail[i]{y}}{\tail[i]{x}}
        \alpha(\st{i}{y})
      \,dy
    \iftoggle{widecol}{\\ &}{\\ &\qquad}
    = \frac{\density[i]{x}}{(\tail[i]{x})^2}\integral{x}{\infty}
        \tail[i]{y}
        \alpha(\st{i}{y})
      \,dy
      + \frac{1}{\tail[i]{x}} \nudgec \dd{x}\integral{x}{\infty}
        \tail[i]{y}
        \alpha(\st{i}{y})
      \,dy
    \iftoggle{widecol}{\\ &}{\\ &\qquad}
    = \hazard{i}{x}\integral{x}{\infty}
        \frac{\tail[i]{y}}{\tail[i]{x}}
        \alpha(\st{i}{y})
      \,dy
      - \alpha(\st{i}{x}),
  \end{align*}
  which with \fref{eq:cost_stage_diff}
  yields $\serve{1}\cost{\st{i}{x}} = - \alpha(\st{i}{x})$,
  as desired.
\end{proof}

\begin{proof}[\proofof{prop:stochastic_first}]
  Clearly $\profit{\single{i}}{r} = \max\{0, r - d\}$
  and $\fair{\single{i}} = d$,
  so by \fref{def:fair} and \fref{thm:composition},
  \begin{align*}
    \fair{\J \seq \single{i}} &= \profitinv{\J}{d} \\
    \fair{\single{i} \seq \J} &= \fair{\J} + d.
  \end{align*}
  \Fref{lem:sjp_properties} implies
  \begin{equation*}
    \profit{\J}{\fair{\single{i} \seq \J}}
    \leq d
    = \profit{\J}{\fair{\J \seq \single{i}}}.
  \end{equation*}
  Again by \fref{lem:sjp_properties},
  $\profit{\J}{r}$ is strictly increasing for $r \geq \fair{\J}$,
  which means $\fair{\single{i} \seq \J} \leq \fair{\J \seq \single{i}}$.
\end{proof}

\section{Spurious Solutions to Optimality Equations}
\label{app:spurious}

Before addressing the specific optimality equations in \fref{sub:optimality},
we give a simple example which illustrates
the main source of spurious solutions and our method for ruling them out.

\subsection{Warmup: Computing Expectation with Dynamic Programming}

Consider a single-stage job with continuous size distribution~$\size{}$,
which has density, tail, and hazard rate functions
$\density{}$, $\tail{}$, and $\hazard{}{}$, respectively.
As usual, we assume $\E{X}$ is finite.
For ease of exposition, we assume $\size{}$ has unbounded support,
but the arguments can be easily modified to handle the bounded case.

Let $\cost{x} = \E{\size{} - x \given \size{} > x}$
be the job's expected remaining size at age~$x$.
We can think of $\cost{}$ as the cost function of the dynamic program
with a single action, namely serving the job,
that incurs cost at continuous rate~$1$ until the job completes.
The optimality equation for this dynamic program is,
by analogy with \fref{def:operators},
\begin{equation}
  \label{eq:opt_trivial}
  \dd{x}\cost{x} + \hazard{}{x}(\cost{{}} - \cost{x}) = -1,
\end{equation}
where $\cost{{}} = 0$.
Solving this equation yields a family of solutions parametrized by $c \in \R$:
\begin{equation*}
  \cost[c]{x}
  = \frac{c}{\tail{x}}
    + \integral{x}{\infty} \frac{\tail{y}}{\tail{x}} \,dy.
\end{equation*}

The true cost function is~$\cost[0]{}$,
so we need a criterion that rules out solutions with $c \neq 0$.
In this case, a sufficient condition is
\begin{equation}
  \label{eq:not_so_fast}
  \lim_{\mathclap{x \to \infty}} \tail{x}|\cost{x}| = 0.
\end{equation}
To see why \fref{eq:not_so_fast} suffices,
observe that because $\E{X}$ is finite,
\begin{equation*}
  \lim_{\mathclap{x \to \infty}} \tail{x}|\cost[0]{x}|
  = \lim_{\mathclap{x \to \infty}}\ \ \integral{x}{\infty} \tail{y} \,dy
  = 0,
\end{equation*}
so $\lim_{x \to \infty} \tail{x}|\cost[c]{x}| = |c|$.
Therefore,
to show that $\cost[0]{}$ is the true cost function,
it suffices to show that the true cost function
must satisfy \fref{eq:not_so_fast},
which we write more succinctly as $\cost{x} = o(1/\tail{x})$.

For the trivial dynamic program in this example,
the simplest way to show $\cost{x} = o(1/\tail{x})$
is to compute the expected remaining size directly,
which removes the need to solve the optimality equation at all.
However, as we will soon see,
we can give asymptotic bounds similar to $o(1/\tail{x})$ for BPOP~Q,
even though directly computing its cost function~$\costQ{}$ is intractable.

\subsection{Ruling Out Spurious Solutions for BPOP~Q}

There are two steps to ruling out spurious solutions to \fref{eq:opt_q}.
We first show that $\costQ{}$ has a ``quadratic'' form,
meaning
\begin{equation}
  \label{eq:quadratic}
  \costQ{\sts}
  = \sum_k \cost[1]{\st{i_k}{x_k}}
    + \sum_{\mathclap{k < \ell}} \cost[2]{\st{i_k}{x_k}, \st{i_\ell}{x_\ell}}
\end{equation}
for some functions $\cost[1]{}$ and~$\cost[2]{}$.
We compute $\cost[1]{}$ in terms of $\cost[2]{}$
and asymptotically bound $\cost[2]{}$ similarly to \fref{eq:not_so_fast},
which rules out all spurious solutions to \fref{eq:opt_q}.

We begin by establishing that $\costQ{}$ has
the quadratic form given by \fref{eq:quadratic}.
We recursively define the \emph{busy period started by job~$k$},
denoted $\mc{B}_k$,
to be the smallest set of jobs containing
\begin{itemize}
\item
  job~$k$ itself and
\item
  all jobs that arrive while serving
  any job in the busy period started by job~$k$.
\end{itemize}

Given any system state,
let $\cost[2]{\st{i_k}{x_k}, \st{i_\ell}{x_\ell}}$ be
the total expected queueing time in the remainder of the busy period
due to jobs in $\mc{B}_k$
waiting while jobs in $\mc{B}_\ell$ are in service,
plus vice versa.
We call this the ``interference'' between $\mc{B}_k$ and~$\mc{B}_\ell$.
To clarify, when there are $n$ jobs from $\mc{B}_k$ in the system,
then each instant serving a job in $\mc{B}_\ell$
counts for $n$ instants of queueing time.

The key observation is that \emph{interference is well defined},
even though it does not account for all jobs in the system.
This is because for the purposes of determining
the interference between $\mc{B}_k$ and~$\mc{B}_\ell$
we can imagine a separate ``$\mc{B}_k$ vs. $\mc{B}_\ell$'' process
which is paused whenever a job not in $\mc{B}_k$ or $\mc{B}_\ell$
is in service.
Because the arrival process is Poisson
and the Gittins policy is an index policy,
the $\mc{B}_k$ vs. $\mc{B}_\ell$ process is unaffected by
other jobs in the system.

The second term on the right-hand side of \fref{eq:quadratic}
accounts for all the remaining expected queueing time from state $\sts$
except for that incurred ``within'' each~$\mc{B}_k$.
We therefore define $\cost[1]{\st{i_k}{x_k}}$
to be the total expected queueing time due to jobs in $\mc{B}_k$
waiting while other jobs in $\mc{B}_k$ are in service,
which we can show is well defined
by an argument similar to that for $\cost[2]{}$ above.
This establishes that $\costQ{}$ has the quadratic from of \fref{eq:quadratic}.

We now compute $\cost[1]{}$ in terms of $\cost[2]{}$.
If a new job~$k'$ arrives when job~$k$ is in state $\st{j}{y}$,
then the interference between $\mc{B}_k$ and $\mc{B}_{k'}$
from that point onward is $\cost[2]{\st{j}{y}, \st{\init}{0}}$.
From this observation we directly compute
\begin{align*}
  \cost[1]{\st{i}{x}}
  &= \integral{x}{\infty}
      \lambda\frac{\tail[i]{y}}{\tail[i]{x}}
      \cost[2]{\st{i}{y}, \st{\init}{0}}
    \,dy \\
  &\qquad + \sum_{\mathclap{j \neq i}}\mkern3mu
      \integral{0}{\infty}
        \lambda q_{ij} \tail[j]{y}
        \cost[2]{\st{j}{y}, \st{\init}{0}}
      \,dy,
\end{align*}
where $q_{ij}\esub$ is the probability that
a job starting in stage~$i$ at some point reaches stage~$j$.

It remains only to characterize $\cost[2]{}$.
It is immediate from applying \fref{eq:quadratic}
to a system state with one job
that $\cost[1]{\st{i}{x}} = \costQ{\st{i}{x}}$.
By applying \fref{eq:opt_q} to a system state with two jobs, we obtain
\begin{equation}
  \label{eq:opt_interference}
  \min_k
  (\serve{k} + \arrive) \cost[2]{\st{i_1}{x_1}, \st{i_2}{x_2}}
  = -1
\end{equation}
Our next step is to turn this into
a single-variable differential equation similar to \fref{eq:opt_trivial}.
To do so, we consider the evolution of the system
assuming that no arrivals or state transitions occur.
We define $x_1(t)$, $x_2(t)$ and $k(t)$ such that
\begin{align*}
  x_1(t) - x_1 &= t - (x_2(t) - x_2) = \integral{0}{t} (2 - k(t)) \,dt \\
  k(t) &= \argmin_k \serve{k} \cost[2]{\st{i_1}{x_1}, \st{i_2}{x_2}}.
\end{align*}
That is, starting from system state
$(\st{i_1}{x_1}, \st{i_2}{x_2})$
and assuming no arrivals or stage transitions,
$x_1(t)$ and $x_2(t)$ are the ages of the stages at time~$t$
and $k(t)$ is the job being served at time~$t$.

Suppose temporarily that stages $i_1$ and $i_2$ are penultimate stages,
meaning they transition only to the final stage~$\fin$.
By \fref{eq:opt_interference},
\begin{equation}
  \label{eq:ddt_interference}
  \dd{t}\cost[2]{t} - \hazard{k(t)}{x_{k(t)}}\cost[2]{t}
  = -(1 + \arrive\cost[2]{t}),
\end{equation}
where we abbreviate
$\cost[2]{t} = \cost[2]{\st{i_1}{x_1(t)}, \st{i_2}{x_2(t)}}$.
We can interpret $\hazard{Y}{t} = \hazard{k(t)}{x_{k(t)}}$
as the hazard rate of a random variable~$Y$,
which is the time of the first stage transition.
It is straightforward to show that
finiteness of $\E{X_1}$ and $\E{X_2}$
implies finiteness of~$\E{Y}$.
We can use~$Y$ to express the possible solutions for $\cost[2]{t}$:
\begin{equation}
  \label{eq:cost_interference}
  \cost[2]{t}
  = \frac{c}{\tail[Y]{t}}
    + \integral{t}{\infty}
      \lambda\frac{\tail[Y]{u}}{\tail[Y]{t}}(1 + \arrive\cost[2]{t})
    \,du.
\end{equation}
Confirming \fref{eq:cost_q_many} entails showing $c = 0$.
By considering the policy that prioritizes new arrivals
above jobs~$1$ and~$2$,
performing the \emph{preemptive last-come, first-served} (PLCFS) policy
\citep{book_harchol-balter}
on the new arrivals,
we obtain the bound
\begin{align*}
  \cost[2]{t}
  \leq \frac{\load{}\E{Y - t \given Y \geq t}}{(1 - \load{})^2}
  = o\biggl(\frac{1}{\tail[Y]{t}}\biggr),
\end{align*}
so $c = 0$ in \fref{eq:cost_interference}, as desired.

We assumed temporarily that stages $i_1$ and~$i_2$ were penultimate stages.
Because there are no cyclic stage transitions,
we can iterate the argument to cover every stage.
The only change is an extra term on
the right-hand side of~\fref{eq:ddt_interference},
which does not substantially change the argument.

\end{document}